\newif\ifprocs
\procstrue
\procsfalse 

\newif\ifeightpage
\eightpagetrue
\eightpagefalse 

\ifprocs
\documentclass{article}
\usepackage{icml2019}
\else
\documentclass[11pt]{article}
\usepackage{fullpage}
\fi
\usepackage{graphicx}
\usepackage{epstopdf}
\usepackage{caption}
\usepackage{color}
\usepackage{amssymb,amsmath}
\usepackage{amsthm}
\usepackage{graphicx}
\usepackage{booktabs}
\usepackage{microtype}

\usepackage{xspace}
\usepackage{subcaption}
\usepackage{calc}

\usepackage{doi}
\usepackage{url}
\hypersetup{colorlinks=true,allcolors=blue}

\newcommand{\commentout}[1]{}

\newcommand{\calC}{{\mathcal C}}

\newcommand{\calY}{\mathcal{Y}}

\newcommand{\calL}{\ensuremath{\mathcal{L}\xspace}}

\newcommand{\R}{\mathbb{R}}

\newcommand{\opt}{\ensuremath{\mathrm{OPT}\xspace}}
\newcommand{\apx}{\ensuremath{\mathrm{APX}\xspace}}

\newcommand{\dist}{\mathsf{dist}}

\DeclareMathOperator{\cost}{cost}
\DeclareMathOperator{\topp}{top}

\newcommand{\ProblemName}[1]{\textsc{#1}}
\newcommand{\OkM}{\ProblemName{Ordered $k$-Median}\xspace}
\newcommand{\pCentrum}{\ProblemName{$p$-Centrum}\xspace}
\newcommand{\kCenter}{\ProblemName{$k$-Center}\xspace}
\newcommand{\kMedian}{\ProblemName{$k$-Median}\xspace}
\newcommand{\kMeans}{\ProblemName{$k$-Means}\xspace}

\newtheorem{theorem}{Theorem}[section]

\newtheorem{lemma}[theorem]{Lemma}
\newtheorem{claim}[theorem]{Claim}

\newtheorem{observation}[theorem]{Observation}
{\theoremstyle{remark} }
{\theoremstyle{definition} }

\newenvironment{proofof}[1]{\begin{proof}[Proof of #1]}{\end{proof}}

\definecolor {processblue}{cmyk}{0.96,0,0,0}
\definecolor{Blue}{rgb}{1.,0.75,0.8}
\definecolor{olive}{rgb}{0.3, 0.4, .1}
\definecolor{fore}{RGB}{249,242,215}
\definecolor{back}{RGB}{51,51,51}
\definecolor{title}{RGB}{255,0,90}
\definecolor{dgreen}{rgb}{0.,0.6,0.}
\definecolor{gold}{rgb}{1.,0.84,0.}
\definecolor{JungleGreen}{cmyk}{0.99,0,0.52,0}
\definecolor{BlueGreen}{cmyk}{0.85,0,0.33,0}
\definecolor{RawSienna}{cmyk}{0,0.72,1,0.45}
\definecolor{Magenta}{cmyk}{0,1,0,0}

\newcommand{\eps}{\epsilon}

\ifprocs
\newcommand{\mycitename}[2]{\citet{#2}}
\else
\newcommand{\mycitename}[2]{{#1}~\cite{#2}}
\fi
\newcommand{\mycite}[1]{\mycitename{}{#1}}

\ifprocs
\else
\title{Coresets for Ordered Weighted Clustering}
\author{
  Vladimir Braverman    \thanks{Johns Hopkins University, USA. This material is based upon work supported in part by the National Science Foundation under
Grants No. 1447639, 1650041 and 1652257, Cisco faculty award, and by the ONR Award N00014-18-1-2364.
Email: \texttt{vova@cs.jhu.edu}
  }
  \and Shaofeng H.-C. Jiang  \thanks{Weizmann Institute of Science.
    This work was partially supported by ONR Award N00014-18-1-2364,
    the Israel Science Foundation grant \#897/13,
    a Minerva Foundation grant, and a Google Faculty Research Award.
    Part of this work was done while was visiting the Simons Institute for the Theory of Computing.
    Email: \texttt{\{shaofeng.jiang,robert.krauthgamer\}@weizmann.ac.il}.
  }
  \and Robert Krauthgamer\footnotemark[2]
  \and Xuan Wu \thanks{Johns Hopkins University, USA. Email: \texttt{wu3412790@gmail.com}}
}
\fi

\begin{document}
\ifprocs
\twocolumn[
\icmltitle{Coresets for Ordered Weighted Clustering}

\icmlsetsymbol{equal}{*}

\begin{icmlauthorlist}
	\icmlauthor{Vladimir Braverman}{equal,jhu}
	\icmlauthor{Shaofeng H.-C. Jiang}{equal,wis}
	\icmlauthor{Robert Krauthgamer}{equal,wis}
	\icmlauthor{Xuan Wu}{equal,jhu}
\end{icmlauthorlist}

\icmlaffiliation{jhu}{Johns Hopkins University, USA. This material is based upon work supported in part by the National Science Foundation under
	Grants No. 1447639, 1650041 and 1652257, Cisco faculty award, and by the ONR Award N00014-18-1-2364.}
\icmlaffiliation{wis}{Weizmann Institute of Science, Israel. This work was partially supported by ONR Award N00014-18-1-2364,
	the Israel Science Foundation grant \#897/13,
	a Minerva Foundation grant, and a Google Faculty Research Award.
	Part of this work was done while was visiting the Simons Institute for the Theory of Computing.}

\icmlcorrespondingauthor{Vladimir Braverman}{vova@cs.jhu.edu}
\icmlcorrespondingauthor{Shaofeng H.-C. Jiang}{shaofeng.jiang@weizmann.ac.il}
\icmlcorrespondingauthor{Robert Krauthgamer}{robert.krauthgamer@weizmann.ac.il}
\icmlcorrespondingauthor{Xuan Wu}{abcd@abc.com}

\icmlkeywords{Coreset, Clustering, Fairness, Ordered Weighted Averaging}

\vskip 0.3in
]
\printAffiliationsAndNotice{\icmlEqualContribution}
\else
\maketitle
\fi

\begin{abstract}
  We design coresets for \OkM, a generalization of classical
  clustering problems such as \kMedian and \kCenter,
  that offers a more flexible data analysis,
  like easily combining multiple objectives
  (e.g., to increase fairness or for Pareto optimization). 
  Its objective function is defined via the
  Ordered Weighted Averaging (OWA) paradigm of Yager (1988), 
  where data points are weighted according to a predefined weight vector,
  but in order of their contribution to the objective (distance from the centers). 

  A powerful data-reduction technique, called a coreset,
  is to summarize a point set $X$ in $\R^d$
  into a small (weighted) point set $X'$,
  such that for every set of $k$ potential centers,
  the objective value of the coreset $X'$ approximates that of $X$
  within factor $1\pm \epsilon$. 
  When there are multiple objectives (weights),
  the above standard coreset might have limited usefulness,
  whereas in a \emph{simultaneous} coreset, 
  which was introduced recently by Bachem and Lucic and Lattanzi (2018), 
  the above approximation holds for all weights (in addition to all centers).
  Our main result is a construction of a simultaneous coreset
  of size $O_{\epsilon, d}(k^2 \log^2 |X|)$ for \OkM.

  To validate the efficacy of our coreset construction
  we ran experiments on a real geographical data set.
  We find that our algorithm produces a small coreset,
  which translates to a massive speedup of clustering computations,
  while maintaining high accuracy for a range of weights. 
\end{abstract}

 \section{Introduction} \label{sec:Intro}

We study data reduction (namely, coresets) for a class of clustering problems,
called ordered weighted clustering, 
which generalizes the classical \kCenter and \kMedian problems.
In these clustering problems, the objective function is computed
by ordering the $n$ data points by their distance to their closest center,
then taking a weighted sum of these distances,
using predefined weights $v_1\geq \cdots \geq v_n \geq 0$. 
These clustering problems can interpolate between
\kCenter (the special case where $v_1=1$ is the only non-zero weight)
and \kMedian (unit weights $v_i=1$ for all $i$),
and therefore offer flexibility in prioritizing points with large service cost,
which may be important for applications 
like Pareto (multi-objective) optimization and fair clustering. 
In general, fairness in machine learning is seeing a surge in interest, 
and is well-known to have many facets. 
In the context of clustering, previous work
such as the fairlets approach of~\mycite{DBLP:conf/nips/Chierichetti0LV17},
has addressed protected classes, which must be identified in advance. 
In contrast, ordered weighted clustering addresses fairness 
towards remote points (which can be underprivileged communities), 
without specifying them in advance.
This is starkly different from many application domains, 
where remote points are considered as outliers (to be ignored)
or anomalies (to be detected), 
see e.g., the well-known survey by~\mycite{DBLP:journals/csur/ChandolaBK09}. 

Formally, we study two clustering problems in Euclidean space $\R^d$.
In both of them, the input is $n$ data points $X\subset \R^d$ (and $k\in[n]$), 
and the goal is to find $k$ centers $C\subset \R^d$
that minimize a certain objective $\cost(X,C)$.
In \OkM, there is a predefined non-decreasing weight vector $v\in \R_+^n$,
and the data points $X=\{x_1,\ldots,x_n\}$ are ordered
by their distance to the centers, i.e., $d(x_1,C) \geq \cdots\geq d(x_n,C)$,
to define the objective 
\begin{equation} \label{eq:costv}
  \cost_v(X,C) := \sum_{i=1}^n v_i \cdot d(x_i,C) ,
\end{equation}
where throughout $d(\cdot,\cdot)$ refers to $\ell_2$ distance, 
extended to sets by the usual convention
$\dist(x, C) := \min_{c \in C} \dist(x, c)$. 
This objective follows the Ordered Weighted Averaging (OWA) paradigm
of \mycite{Yager88}, 
in which data points are weighted according to a predefined weight vector,
but in order of their contribution to the objective.
The \pCentrum problem is the special case 
where the first $p$ weights equal $1$ and the rest are $0$, 
denoting its objective function by $\cost_p(X,C)$.
Observe that this problem already
includes both \kCenter (as $p=1$) and \kMedian (as $p=n$).

A powerful data-reduction technique, called a \emph{coreset},
is to summarize a large point set $X$ into a (small) multiset $X'$,
that approximates well a given cost function (our clustering objective) 
for every possible candidate solution (set of centers). 
More formally, $X'$ is an $\eps$-coreset of $X$
for clustering objective $\cost(\cdot,\cdot)$
if it approximates the objective within factor $1\pm\eps$, i.e., 
\[
  \forall C\subset \R^d, |C|=k, 
  \quad
  \cost(X', C) \in (1\pm \eps) \cost (X, C) .
\]
The \emph{size} of $X'$ is the number of \emph{distinct} points in it.\footnote{A common alternative definition is that
  $X'$ is as a set with weights $w:X'\to\R_+$, which represent multiplicities,
  and then size is the number of non-zero weights.
  This would be more general if weights are allowed to be fractional, 
  but then one has to extend the definition of $\cost(\cdot,\cdot)$ accordingly.
}
The above notion, sometimes called a strong coreset, was proposed by \mycite{HM04},
following a weaker notion of \mycite{AHV04}. 
In recent years it has found many applications,
see the surveys of \mycite{AHV05}, \mycite{Phillips16} and \mycite{MS18},
and references therein.

The above coreset definition readily applies to ordered weighted clustering.
However, a standard coreset is constructed for a specific clustering objective,
i.e., a single weight vector $v\in\R_+^n$, which might limit its usefulness. 
The notion of a \emph{simultaneous} coreset,
introduced recently by \mycite{BachemLL18}, 
requires that all clustering objectives are preserved,
i.e., the $(1+\eps)$-approximation holds for all weight vectors 
in addition to all centers. 
This ``simultaneous'' feature is valuable in data analysis, 
since the desired weight vector might be application and/or data dependent,
and thus not known when the data reduction is applied. 
Moreover, since ordered weighted clustering includes classical clustering,
e.g., \kMedian and \kCenter as special cases,
all these different analyses may be performed on a single simultaneous coreset.

\subsection{Our Contribution}
\label{sec:contribution}

Our main result is (informally) stated as follows.
To simplify some expressions,
we use $O_{\eps,d}(\cdot)$ to suppress factors depending only on $\eps$ and $d$.
The precise dependence appears in the technical sections. 

\begin{theorem}[informal version of Theorem~\ref{thm:okm}]
  \label{thm:okm_informal}
  There exists an algorithm that,
  given an $n$-point data set $X \subset \R^d$ and $k\in[n]$,
  computes a simultaneous $\eps$-coreset of size $O_{\eps,d}(k^2 \log^2{n})$
  for \OkM.
\end{theorem}

Our main result is built on top of a coreset result for \pCentrum 
(the special case of \OkM in which the weight vector is $1$ in the first $p$ components and $0$ in the rest).
For this special case, we have an improved size bound,
that avoids the $O(\log^2 n)$ factor, stated as follows.
Note that this coreset is for a single value of $p$ (and not simultaneous).

\begin{theorem}[informal version of Theorem~\ref{thm:pcent}]
  \label{thm:pcent_informal}
  There exists an algorithm that,
  given an $n$-point data set $X \subset \R^d$ and $k,p\in [n]$,
  computes an $\eps$-coreset of size $O_{\eps, d}(k^2)$ for $\pCentrum$.
\end{theorem}

The size bounds in the two theorems are nearly tight. 
The dependence on $n$ in Theorem~\ref{thm:okm_informal} is unavoidable,
because we can show that the coreset size has to be $\Omega(\log n)$,
even when $k=d=1$ (see Theorem~\ref{thm:lb}).
For both Theorem~\ref{thm:okm_informal} and Theorem~\ref{thm:pcent_informal}, the hidden dependence on $\eps$ and $d$ is $(\frac{1}{\eps})^{d+O(1)}$.
This factor matches known lower bounds [D.~Feldman, private communication]
and state-of-the-art constructions of coresets for \kCenter
(which is a special case of \OkM)~\cite{DBLP:journals/algorithmica/AgarwalP02}.

A main novelty of our coreset is that it preserves the objective for all weights ($v\in\R_+^n$ in the objective function) simultaneously. 
It is usually easy to combine coresets for two data sets,
but in general it is not possible to combine coresets for two different objectives.
Moreover, even if we manage to combine coresets for two objectives,
it is still nontrivial to achieve a small coreset size for
infinitely many objectives (all possible weight vectors $v\in\R_+^n$).
See the overview in Section~\ref{sec:intro_tech}
for more details on the new technical ideas needed to overcome these obstacles.

We evaluate our algorithm on a real 2-dimensional geographical data set with about 1.5 million points.
We experiment with the different parameters for coresets of \pCentrum,
and we find out that the empirical error is always far lower
than our error guarantee $\eps$. 
As expected, the coreset is much smaller than the input data set,
leading to a massive speedup (more than 500 times) in the running time
of computing the objective function.
Perhaps the most surprising finding is that a single \pCentrum coreset 
(for one ``typical'' $p$) empirically serves as a simultaneous coreset,
which avoids the more complicated construction
and the dependence on $n$ in Theorem~\ref{thm:okm_informal},
with a coreset whose size is only 1\% of the data set. 
Overall, the experiments confirm that our coreset is practically efficient,
and moreover it is suitable for data exploration,
where different weight parameters are needed.

\subsection{Overview of Techniques}
\label{sec:intro_tech}

We start with discussing Theorem~\ref{thm:pcent_informal}
(which is a building block for Theorem~\ref{thm:okm_informal}). 
Its proof is inspired by~\mycite{har2007smaller},
who constructed coresets for \kMedian clustering in $\R^d$
by reducing the problem to its one-dimensional case. 
We can apply a similar reduction,
but the one-dimensional case of \pCentrum is significantly different
from \kMedian.
One fundamental difference is that the objective counts only the $p$ largest distances, hence the subset of ``contributing'' points depends on the center.
We deal with this issue by introducing a new bucketing scheme
and a charging argument that relates the error to the $p$ largest distances.
See Section~\ref{sec:toy} for more details. 

The technical difficulty in Theorem~\ref{thm:okm_informal} is two-fold:
how to combine coresets for two different weight vectors,
and how to handle infinitely many weight vectors. 
The key observation is that every \OkM objective can be represented
as a linear combination of \pCentrum objectives (see Lemma~\ref{centogen}).
Thus, it suffices to compute a simultaneous coreset for \pCentrum
for all $p \in [n]$.
We achieve this by ``combining'' the individual coresets for all $p \in [n]$,
while crucially utilizing the special structure of our construction 
of a \pCentrum coreset, but unfortunately losing an $O(\log n)$ factor in the coreset size. 
In the end, we need to ``combine'' the $n$ coresets for all $p\in[n]$, 
but we can avoid losing an $O(n)$ factor by discretizing the values of $p$,
so that only $O(\log n)$ coresets are combined,
The result is a simultaneous coreset of size $O_{\eps, d}(\log^2 n)$,
see Section~\ref{sec:simul} for more details.

\subsection{Related Work}
\label{sec:related}

The problem of constructing strong coresets for \kMeans, \kMedian, and other objectives has received significant attention from the research community \cite{Feldman:2010:CSH:1873601.1873654,Feldman:2011:UFA:1993636.1993712,DBLP:conf/soda/LangbergS10,Badoiu:2002:ACV:509907.509947,Chen:2009:CKK:1655307.1655320}.
For example, \mycite{HM04}
designed the first strong coreset for \kMeans.
\mycite{Feldman:2013:TBD:2627817.2627920} provided coresets for \kMeans,
PCA and projective clustering that are independent of the dimension.
Recently, \mycite{DBLP:conf/focs/SohlerW18} generalized the results of \mycite{Feldman:2013:TBD:2627817.2627920} and obtained strong coresets for \kMedian and for subspace approximation that are independent of the dimension $d$.

\OkM and its special case \pCentrum generalize \kCenter
and are thus APX-hard even in $\mathbb{R}^2$ \cite{MS84}. 
However, \pCentrum may be solved optimally in polynomial time for special cases such as lines and trees~\cite{tamir01}. The first provable approximation algorithm for \OkM was proposed by~\mycite{aouad2018ordered}, and they gave $2$-approximation for trees and $O(\log n)$-approximation for general metrics. The approximation ratio for general metrics was drastically improved to $38$
by~\mycite{byrka2018constant}, improved to $18+\epsilon$
by~\mycite{chakrabarty2017interpolating}, and finally a $(5+\epsilon)$-approximation was obtained very recently by~\mycite{chakrabarty2018approximation}.

Previous work on fairness in clustering
has followed the disparate impact doctrine of \mycite{FFMSV15},
and addressed fairness with respect to protected classes,
where each cluster in the solution should fairly represent every class. 
\mycite{DBLP:conf/nips/Chierichetti0LV17}
have designed approximation algorithms for \kCenter and \kMedian,
and their results were refined and extended
by~\mycite{RS18} and~\mycite{BCN19}.
Recent work by \mycite{SSS18} designs coresets for fair \kMeans clustering.
However, these results are not applicable to ordered weighted clustering.

 \section{Preliminaries}
\label{sec:prelims}

Throughout this paper we use capital letters other than $I$ and $J$
to denote finite subsets of $\R^d$.
We recall some basic terminology from \cite{har2007smaller}. 
For a set $Y\subset \mathbb{R}$,
define its \emph{mean point} to be
\begin{equation} \label{eq:mu}
  \mu(Y):=\frac{1}{|Y|}\sum_{y\in Y} y,
\end{equation}
and its \emph{cumulative error} to be
\begin{equation} \label{eq:delta}
  \delta(Y):=\sum_{y\in Y} |y-\mu(Y)|.
\end{equation}
Let $I(Y):=[\inf Y,\sup Y]$ denote the smallest closed interval containing $Y$.
The following facts from \cite{har2007smaller} will be useful in our analysis.

\begin{lemma} \label{basicfact}
  For every $Y\subset \mathbb{R}$ and $z\in \mathbb{R}$, 
\begin{itemize} 
\item $\sum_{y\in Y} \Big| |z-y|-|z-\mu(Y)| \Big| \leq \delta(Y)$; and 
\item if $z\notin I(Y)$ then $\sum_{y\in Y} |y-z|=|Y|\cdot |\mu(Y)-z|$.
\end{itemize}
\end{lemma}

It will be technically more convenient to treat a coreset
as a point set $X'\subset\R^d$ associated with integer weights $w:X'\to\mathbb{N}$, 
which is equivalent to a multiset (with weights representing multiplicity),
and thus the notation of $\cost_v(X', C)$ in \eqref{eq:costv} is well-defined.
(These weights $w$ are unrelated to the predefined weights $\{v_i\}$.)
While our algorithm always produces $X'$ with integral weights $w$,
our proof requires fractional weights during the analysis,
and thus we extend \eqref{eq:mu} and \eqref{eq:delta} 
to a point set $Y$ with weights $w: Y \rightarrow \R_+$ by defining
\begin{align*}
  \mu_w(Y)
  &:= \frac{1}{\sum_{y \in Y}{w(y)}}\sum_{y\in Y}{ w(y) \cdot y },
  \\
  \delta_w(Y)
  &:= \sum_{y \in Y}{w(y) \cdot |y - \mu(Y)|}.
\end{align*}

We will use the fact that in one-dimensional Euclidean space, 
\pCentrum can be solved (exactly) in polynomial time by dynamic programming,
as shown by~\mycite{tamir01}.

\begin{lemma}[\cite{tamir01}]\label{exactsol}
  There is a polynomial-time algoritm that,
  given a set of one-dimensional points $X=\{x_1,\ldots,x_n\}\subset\R$
  and parameters $k,p\in[n]$,
  computes a set of $k$ centers $C\subset\R^d$
  that minimizes $\cost_p(X,C)$. 
\end{lemma}

 \section{The Basic Case: \pCentrum for $k=d=1$ (one facility in one-dimensional data)}
\label{sec:toy}

In this section we illustrate our main ideas by constructing a coreset for
\pCentrum in the special case of one facility in one-dimensional Euclidean space (i.e., $k=d=1$).
This is not a simultaneous coreset, but rather for a single $p$.
The key steps of our construction described below
will be repeated, with additional technical complications,
also in the general case of \pCentrum, i.e., $k$ facilities in dimension $d$.

\ifeightpage
We will need two technical lemmas, whose proof can be found in the full version.
\else
We will need two technical lemmas,
whose proofs appear in Section~\ref{sec:basic_proofs}.
\fi
The first lemma bounds $\delta(Y)$
by the cost of connecting $Y$ to an arbitrary point outside $I(Y)$
(which in turn is part of the objective in certain circumstances).

\begin{lemma}\label{cum}
  Let $Y\subset \R$ be a set with (possibly fractional) weights $w:Y\to\R_+$. 
  Then for every $z\in\R$ such that $z\not\in I(Y)$ or $z$ is an endpoint of $I(Y)$,   $$
  \delta_w(Y)\leq 2\sum_{y\in Y} w(y) \cdot |y-z|.
  $$
\end{lemma}

Recall that $k=1$, hence the cost in an instance of \pCentrum
is the sum of the $p$ largest distances to the center.
In the analysis of our coreset it will be useful to replace some points
of the input set $X$ with another set $Y$.
The second lemma will be used to bound the resulting increase in the cost;
it considers two sequences, denoted $X$ and $Y$, of the connection costs,
and bounds the difference between the sum of the $p$ largest values in $X$
and that in $Y$ by a combination of $\ell_{\infty}$ and $\ell_1$ norms.

\begin{lemma}\label{lonelinf}
  Let $X=(x_1,\ldots,x_n)$ and $Y=(y_1,\ldots,y_n)$ be two sequences of real numbers.
  Then for all $S\subseteq [n]$,
  $$
  |\topp_p(X)-\topp_p(Y)|
  \leq
  p\max_{i\in S} |x_i-y_i|+\sum_{i\in [n]\setminus S} |x_i-y_i|,
  $$
  where $\topp_p(Z)$ is the sum of the $p$ largest numbers in $Z$.
\end{lemma}

\paragraph{Outline of the Coreset Construction}
In the context of a one-dimensional point set $X\subset \R$, 
the term \emph{interval} will mean a subset of $X$
that spans a contiguous subsequence under a fixed ordering of the points, 
i.e., a subset $\{x_i,\ldots,x_j\}$
when the points in $X$ are ordered as $x_1 \leq \ldots \leq x_n$.
Informally, our coreset construction works as follows. 
First, use Lemma~\ref{exactsol} to find an optimal center $y^*$,
its corresponding optimal cost $\opt$,
and a subset $P\subset X$ of size $|P|=p$ that contributes to the optimal cost.
Then partition the data into three intervals, namely $X=L\cup R\cup Q$,
as follows.
Points from $P$ that are smaller or equal to $y^*$ are placed in $L$, 
points from $P$ that are larger than $y^*$ are placed in $R$,
and all other points are placed in $Q=X\setminus P$.
Now split $L$, $Q$ and $R$ into sub-intervals,
in a greedy manner that we describe below,
and represent the data points in each sub-interval
by adding to the coreset a single point, whose weight
is equal to the number of data points it replaces.
See Figure~\ref{fig:toy} for illustration.

\begin{figure*}[ht]
	\captionsetup{font=small}
	\centering
	\begin{subfigure}[b]{0.48\linewidth}
		\includegraphics[width=0.95\linewidth]{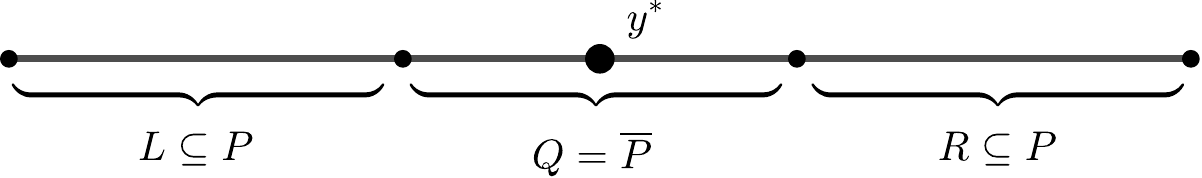}
	\end{subfigure}
	\quad
	\begin{subfigure}[b]{0.48\linewidth}
		\includegraphics[width=0.95\linewidth]{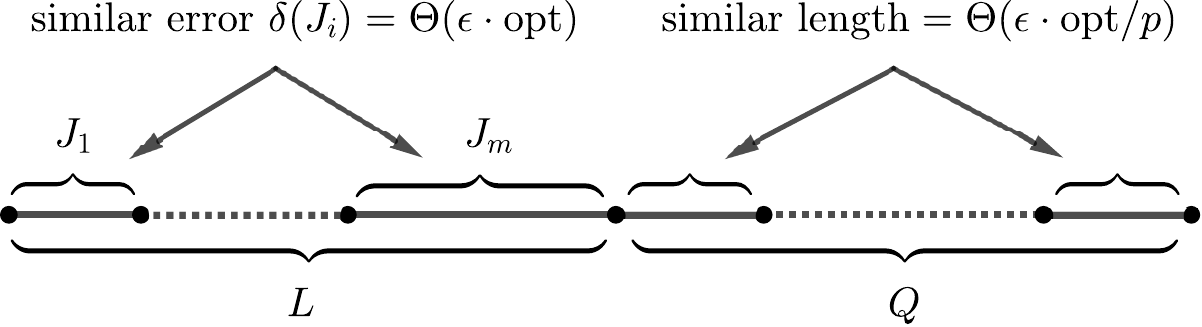}
	\end{subfigure}
	\caption{Coreset construction for \pCentrum with $k=1$ facilities in dimension $d=1$.
		The left figure depicts the partition of the data into $X=(L\cup R)\cup Q$,
		where $P=L\cup R$ contains the $p$ furthest points from an optimal center $y^*$.
		The right figure shows the different manners of splitting $L$ and $Q$ into intervals.
	}
	\hrulefill
	\label{fig:toy}
\end{figure*}

To split $L$ into sub-intervals,
scan its points from the smallest to the largest
and pack them into the same sub-interval $J$ as long as their cumulative error $\delta(J)$
is below a threshold set to $\Theta({\varepsilon \cdot \opt})$. 
This ensures, by Lemma~\ref{cum},
a lower bound on their total connection cost to the optimal center $y^*$,
which we use to upper bound the number of such intervals
(which immediately affects the size of the coreset)
by $O\left({1\over \epsilon}\right)$.
The split of $R$ is done similarly.
To split $Q=X\setminus P$,
observe that the distance from every $q\in Q$ to the center $y^*$
is less than $\frac{\opt}{p}$,
hence the diameter of $Q$ is less than $\frac{2\opt}{p}$,
and $Q$ can be partitioned into $O(\frac{1}{\varepsilon})$
sub-intervals of length $O(\frac{\varepsilon\opt}{p})$.
Observe that the construction for $Q$ differs from that of $L$ and $R$.

Let $D$ denote the coreset resulting from the above construction.
To prove that the resulting coreset has the desired error bound
for every potential center $y\in\R$, 
we define an intermediate set $Z$ that contains a mix of points from $X$ and $D$.
We stress that $Z$ depends on the potential center $y\in\R$, 
which is possible because $Z$ is used only in the analysis. 
The desired error bound follows by bounding both
$|\cost(Z, y) - \cost(X, y)|$ and $|\cost(Z, y) - \cost(D, y)|$,
(here we use Lemma~\ref{lonelinf}), 
and applying the triangle inequality.

\paragraph{Detailed Construction and Coreset Size}

We now give a formal description of our coreset construction.
Let $X=\{x_1,\ldots,x_n\}\subset \R$ be the input data set,
and recall that $\cost_p(X,y)$ for a point $y\in \R$
is the sum of the $p$ largest numbers in $\{|x_1-y|,\ldots,|x_n-y|\}$.
Denote the optimal center by $y^*:=\mathrm{argmin}_{y\in \R} \cost_p(X,y)$, and the corresponding optimal cost by $\opt:=\cost_p(X,y^*)$.
By Lemma~\ref{exactsol},
$y^*$ and $\opt$ can be computed in polynomial time.
Next, sort $X$ by distances to $y^*$.
For simplicity, we shall assume the above notation for $X$
is already in this sorted order, i.e.,
$|x_1-y^*|\geq \cdots\geq |x_n-y^*|$.
Thus, $\cost_p(X,y^*)=\sum_{i=1}^p |x_i-y^*|$.

Let $P := \{x_1,\ldots,x_p\}$, $L := \{x_i\leq y^*: x_i\in P\}$, $R := \{x_i> y^*: x_i\in P\}$ and $Q := X\setminus P$.
By definition, $X$ is partitioned into $L$, $Q$ and $R$,
which form three intervals located from left to right.
We now wish to split $L$, $Q$ and $R$ into sub-intervals,
and then we will add to $D$ the mean of the points in each sub-interval,
with weight equal to the number of such points. 

Split $L$ into sub-intervals from left to right greedily,
such that the cumulative error of each interval $J$ does not exceed $\frac{2\varepsilon \cdot\opt}{21}$,
and each sub-intervals is maximal, i.e., the next point cannot be added to it.
Split $R$ into sub-intervals similarly but from right to left.
We need to bound the number of sub-intervals produced in this procedure.

For sake of analysis,
we consider an alternative split of $L$ that is fractional, 
i.e., allows assigning a point fractionally to multiple sub-intervals, 
say $1/3$ to the sub-interval to its left
and $2/3$ to the sub-interval to its right.
The advantage of this fractional split is that all but the last sub-interval
have cumulative error \emph{exactly}  $\frac{2\varepsilon \cdot\opt}{21}$. 
We show in Lemma~\ref{lem:comp_frac_int} that the number of sub-intervals  produced in the original integral split
is at most twice that of the fractional split,
and thus it would suffice to bound the latter by $O(\frac{1}{\varepsilon})$.

\begin{lemma}
  \label{lem:comp_frac_int}
  The number of sub-intervals in the integral split is at most twice than that of the fractional split.
\end{lemma}

\ifeightpage
\begin{proof}
	The proof can be found in the full version.
\end{proof}
\else
\begin{proof}
It suffices to show that for every $t$,
the first $2t$ sub-intervals produced by the integral partitioning contain at least as many points as the first $t$ sub-intervals produced by the fractional partitioning.
For this, the key observation is that in a fractional split,
only the two endpoints of a sub-interval may be fractional,
because the cumulative error of a singleton set is $0$.

We prove the above by induction on $t$.
The base case $t=1$ follows from these observations, since the fractional sub-interval may be broken into two integral sub-intervals,
each with cumulative error at most $\frac{2\varepsilon \cdot\opt}{21}$.
Suppose the claim holds for $t-1$ and let us show that it holds for $t$.
Since the cumulative error is monotone in adding new points, we may assume the first $t-1$ sub-intervals from fractional split contain as many points as the first $2(t-1)$ sub-intervals from integral split.
Now similarly to the base case, the $t$-th fractional interval may be broken into two integral sub-intervals, and this proves the inductive step.
\end{proof}
\fi

To see that the number of sub-intervals produced by a fractional partitioning  of $L$ is $O(\frac{1}{\varepsilon})$, we use Lemma~\ref{cum}.
Suppose there are $m$ such sub-intervals $J_1,...,J_m$.
We can assume that the first $m-4$ of them do not contain $y^*$
and have cumulative error at least $\frac{2\varepsilon\cdot \opt}{21}$,
because at most two sub-intervals can contain $y^*$, 
and at most one sub-interval from each of $L$ and $R$ may have
cumulative error less than $\frac{2\varepsilon\cdot \opt}{21}$. 
By Lemma~\ref{cum} and the fact that $y^*$ is not in the first $i\le m-4$  sub-intervals,
\[
  \opt
  \geq \sum_{i=1}^{m-4}\sum_{x:x\in J_i} |x-y^*|
  \geq  \frac{1}{2}\sum_{i=1}^{m-4} \delta_w(J_i)=(m-4)\frac{\varepsilon\cdot \opt}{21} .
\]
Thus $m=O(\frac{1}{\varepsilon})$,
and by Lemma~\ref{lem:comp_frac_int} a similar bound holds also for
the number of sub-intervals in the integral split of $L$ and of $R$. 

Now split $Q$ greedily into maximal sub-intervals of length not larger than $\frac{\varepsilon\cdot \opt}{3p}$.
Since $\max_{q\in Q} |q-y^*|\leq |x_p-y^*|\leq \frac{\opt}{p}$,
the length of $I(Q)$ is at most $\frac{2\opt}{p}$,
and we conclude that $Q$ is split into
at most $\frac{3}{\varepsilon}+1$ sub-intervals.

Finally, construct the coreset $D$ from the sub-intervals,
by adding to $D$ the mean of each sub-interval in $D$,
with weight that is the number of points in this sub-interval.
Since the total number of sub-intervals is $O(\frac{1}{\eps})$,
the size of the coreset $D$ is also bounded by $O(\frac{1}{\eps})$.

\paragraph{Coreset Accuracy}
To prove that $D$ is an $\eps$-coreset for $X$,
fix a potential center $y\in \mathbb{R}$ and let us prove that $|\cost_p(D,y)-\cost_p(X,y)|\leq \varepsilon\cdot\opt$,
where we interpret $D$ as a multi-set. 
Let $P_1\subseteq X$ denote the set of $p$ points in $X$ that are farthest from $y$.
Now define an auxiliary set $Z := \{z_1,\ldots,z_n\}$,
as follows.
For each $i \in [n]$, let $X_i\subset X$ be the sub-interval containing $x_i$
in the construction of the coreset
(recall it uses the optimal center $y^*$ and \emph{not} $y$), 
and let $\pi(x_i)=\mu(X_i)$ be its representative in the coreset $D$.
Now if (a) $i\leq p$;
(b) $y\not\in X_i$; and
(c) $P_1\cap X_i$ is either empty or all of $X_i$;
then let $z_i := \pi(x_i)$.
Otherwise, let $z_i:=x_i$.

We now aim to bound $|\cost_p(Z,y)-\cost_p(D,y)|$
using Lemma~\ref{lonelinf} with $S=\{p+1,...,n\}$.
Consider first some $i\in S$ (i.e., $i>p$). Then
\begin{align}
|d(z_i,y)-d(\pi(x_i),y)|
&\leq |z_i-\pi(x_i)| \nonumber \\
&=|x_i-\pi(x_i)|
\leq \frac{\varepsilon\cdot \opt}{3p}. \label{k1inf}
\end{align}
Consider next $i\notin S$ (i.e., $i\le p$).
We can have $z_i\neq\pi(x_i)$ only if $y\in X_i$
or if $P_1\cap X_i$ is neither empty nor all of $X_i$.
This can happen for at most $7$ distinct sub-intervals $X_i$,
because the former case can happen for at most $3$ sub-intervals $X_i$
(by a simple case analysis of how many sub-intervals might have an endpoint
at $y$, e.g., two from $L$, or one from each of $L,R,Q$) 
and because $P_1$ is contained in $2$ intervals (to the left and right of $y$),
and each of them can intersect at most $2$ distinct sub-intervals $X_i$
without containing all of $X_i$.
We obtain
\begin{align}
\sum_{i=1}^p & |d(z_i,y)-d(\pi(x_i),y)| = \nonumber\\
&=\sum_{i \in [p]: z_i \neq \pi(x_i)} |d(x_i,y)-d(\pi(x_i),y)| \label{q1}\\
&\leq \sum_{X_i : i \in [p], (y \in X_i) \lor (P_1 \cap X_i \neq\emptyset,X_i)} \delta(X_i) \label{inq2}\\
&\leq 7\cdot \frac{2\varepsilon \cdot\opt}{21}=\frac{2\varepsilon \cdot\opt}{3} \label{k1l1},
\end{align}
where~\eqref{inq2} is by Lemma~\ref{basicfact},
and~\eqref{k1l1} is by the fact that these $X_i$ are from $L$ or $R$ 
(recall $i \le p$) 
and thus have a bounded cumulative error.

Applying Lemma~\ref{lonelinf} to our $S=\{p+1,...,n\}$
together with~\eqref{k1inf} and~\eqref{k1l1}, we obtain
$$
  |\cost_p(Z,y)-\cost_p(D,y)|
  \leq p\cdot \frac{\varepsilon\cdot \opt}{3p}+\frac{2\varepsilon \cdot\opt}{3}
  =\varepsilon \cdot \opt.
$$

Lastly, we need to prove that $\cost_p(Z,y)=\cost_p(X,y)$.
We think of $Z$ as if it is obtained from $X$ by replacing
each $x_i$ with its corresponding $z_i=\pi(x_i)=\mu(X_i)$.
We can of course restrict attention to indices where $z_i\neq x_i$,
which happens only if all three requirements (a)-(c) hold.
Moreover, whenever this happens for point $x_i$,
it must happen also for all points in the same sub-interval $X_i$,
i.e., every $x_j\in X_i$ is replaced by $z_j=\pi(x_j)=\mu(X_i)$.
By requirement (c), $X_i$ is either disjoint from $P_1$ or contained in $P_1$.
In the former case, points $x_j\in X_i$ do not contribute to $\cost_P(X,y)$
because they are not among the $p$ farthest points,
and then replacing all $x_j\in X_i$ with $z_j=\mu(X_i)$ would maintain this,
i.e., the corresponding points $z_j$ do not contribute to $\cost_p(Z,y)$.
In the latter case, the points in $X_i$ contribute to $\cost_P(X,y)$
because they are among the $p$ farthest points,
and replacing every $x_j\in X_i$ with $z_j=\mu(X_i)$ would maintain this,
i.e., the corresponding points $z_j$ contribute to $\cost_p(Z,y)$.
Moreover, their total contribution is the same
because using requirement (b) that $y\notin X_i$,
we can write their total contribution as
$
  \sum_{x_j\in X_i} d(x_j,y)
  = |X_i| \cdot d(\mu(X_i),y)
  =\sum_{x_j\in X_i} d(\pi(x_j),y)
$.

\ifeightpage
\else
\subsection{Proofs of Technical Lemmas}
\label{sec:basic_proofs}

\begin{lemma}[restatement of Lemma~\ref{cum}]
  Let $Y\subset \R$ be a set with (possibly fractional) weights $w:Y\to\R_+$. 
  Then for every $z\in\R$ such that $z\not\in I(Y)$ or $z$ is an endpoint of $I(Y)$,   $$
  \delta_w(Y)\leq 2\sum_{y\in Y} w(y) \cdot |y-z|.
  $$
\end{lemma}
\begin{proof}
  Assume w.l.o.g. that $z$ is to the left of $I(Y)$, i.e., $z\leq \inf_{y\in Y}y $.
  Partition $Y$ into $Y_L=\{y\in Y: y\leq \mu_w(Y)\}$ and $Y_R=Y\setminus Y_L$. Define $w_L:=\sum_{x\in Y_L} w(x)$ and $w_R:=\sum_{x\in Y_R} w(x)$.
  Since
$$(w_L+w_R)\cdot \mu_w(Y)=\sum_{y\in Y} w(y)\cdot y=\sum_{y\in Y_L} w(y)\cdot y+\sum_{y\in Y_R} w(y)\cdot y,$$
we have that
$$
\sum_{y\in Y_L} w(y)(\mu_w(Y)-y)=\sum_{y\in Y_R} w(y)(y-\mu_w(Y)).
$$
For every $y\in Y_L$ we actually have $z\leq y\leq \mu_w(Y)$,
and we conclude that
\begin{align*}
2\sum_{y\in Y} w(y)(y-z)
&=2(w_L+w_R)\cdot (\mu_w(Y)-z)\\
&\geq 2 w_L \cdot (\mu_w(Y)-z)\\
&\geq 2\sum_{y\in Y_L} w(y)(\mu_w(Y)-y)\\
&=\sum_{y\in Y_L} w(y)(\mu_w(Y)-y)+\sum_{y\in Y_R} w(y)(y-\mu_w(Y))\\
&=\delta_w(Y).
\end{align*}
\end{proof}

\begin{lemma}[restatement of Lemma~\ref{lonelinf}]
  Let $X=(x_1,\ldots,x_n)$ and $Y=(y_1,\ldots,y_n)$ be two sequences of real numbers.
  Then for all $S\subseteq [n]$,
  $$
  |\topp_p(X)-\topp_p(Y)|
  \leq
  p\max_{i\in S} |x_i-y_i|+\sum_{i\in [n]\setminus S} |x_i-y_i|,
  $$
  where $\topp_p(Z)$ is the sum of the $p$ largest numbers in $Z$.
\end{lemma}

\begin{proof}
For all $T\subseteq [n],|T|=p$, \begin{align*}
\left|\sum_{i\in T} (x_i-y_i)\right|
&\leq \sum_{i\in T\cap S} |x_i-y_i|+\left|\sum_{i\in T\setminus S}(x_i-y_i)\right| \\
&\leq p\max_{i\in S} |x_i-y_i|+\sum_{i\in [n]\setminus S} |x_i-y_i|.
\end{align*}
Now let $X_1\subseteq [n]$ be the set of indices of the $p$ largest numbers in $X$, then by the above inequality
\begin{align*}
\topp_p&(X)
=\sum_{i\in X_1} x_i\leq \sum_{i\in X_1} y_i+\bigg{|}\sum_{i\in X_1} (x_i-y_i)\bigg{|}\\
&\leq \topp_p(Y)+p\max_{i\in S} |x_i-y_i|+\sum_{i\in [n]\setminus S} |x_i-y_i|.
\end{align*}
By symmetry, the same upper bound holds also for $\topp_p(Y) - \topp_p(X)$,
and the lemma follows.
\end{proof}
\fi

 \section{Simultaneous Coreset for \OkM}
\label{sec:simul}

In this section we give the construction of a simultaneous coreset for \OkM
on data set $X \subset \mathbb{R}^d$ (Theorem~\ref{thm:okm}),
which in turn is based on a coreset for \pCentrum (Theorem~\ref{thm:pcent}).
In both constructions, we reduce the general instance in $\mathbb{R}^d$
to an instance $X'$ that lies on a small number of lines in $\mathbb{R}^d$. 

The reduction is inspired by a projection procedure of~\mycite{har2007smaller},
that goes as follows.
We start with an initial centers set $C$, and then for each center $c \in C$,
we shoot $O(\frac{1}{\eps})^d$ lines from center $c$ to different directions,
and every point in $X$ is projected to its closest line.
The projection cost is bounded because the number of lines shot from each center is large enough to accurately discretize all possible directions.
The details appear in Section~\ref{sec:highdim}.

For the projected instance $X'$, we construct a coreset for each line in $X'$ using ideas similar to the case $d = k = 1$,
which was explained in Section~\ref{sec:toy}.
However, the error of the coreset cannot be bounded line by line,
and instead, we need to address the cost globally for all lines altogether,
see Lemma~\ref{lem:1d} for the formal analysis.
Finally, to construct a coreset for \pCentrum in $\R^d$,
the initial centers set $C$ for the projection procedure
is picked using some polynomial-time $O(1)$-approximation algorithm,
such as by~\mycite{chakrabarty2018approximation}. 
A coreset of size $O_{\eps, d}(k^2)$ is obtained by combining the projection procedure with Lemma~\ref{lem:1d}.

To deal with the infinitely many potential weights in the simultaneous coreset for \OkM, the key observation is that it suffices to construct a simultaneous coreset for \pCentrum for $O(\frac{\log n}{\eps})$ different value of $p$,
and then ``combine'' the corresponding \pCentrum coresets.
An important structural property of the \pCentrum coreset is that it is formed by mean points of some sub-intervals. This enables us to ``combine'' coresets for \pCentrum by ``intersecting'' all their sub-intervals
into even smaller intervals.
However, this idea works only when the sub-intervals are defined on the same set of lines, which were generated by the projection procedure. 
To resolve this issue, we set the centers set $C$ in the projection procedure
to be the union of all centers needed for \pCentrum
in all the $O(\log n)$ values of $p$. 
Since the combination of the coresets for \pCentrum yields even smaller sub-intervals, the error analysis for the individual coreset for \pCentrum still carries on.
The size of the simultaneous coreset is $O(\log^2 n)$-factor larger
than that for (a single) \pCentrum,
because we combine $O(\log n)$ coresets for \pCentrum,
and we use $O(\log n)$ times more centers in the projection procedure.
The detailed analysis appears in Section~\ref{sec:order}.

 \subsection{Coreset for \pCentrum on Lines in $\mathbb{R}^d$}
\label{sec:1d}

Below, we prove the key lemma that bounds the error of the coreset for \pCentrum for a data set that may be represented by lines. The proof uses the idea introduced for the $k = d = 1$ case in Section~\ref{sec:toy}. In particular, we define an intermediate (point) set $Z$ to help compare the costs between the coreset and the true objective. The key difference from Section~\ref{sec:toy} in defining $Z$ is that the potential centers might not be on the lines, so extra care should be taken.
Moreover, we use a global cost argument to deal with multiple lines in $X$.

We also introduce parameters $s$ and $t_l$ in the lemma. These parameters are to be determined with respect to the initial center set $C$ in the projection procedure, and eventually we want $(s + \sum_{l \in \calL}{t_l})$ to be $O(\opt_p)$ where $\opt_p$ is the optimal for \pCentrum.
We introduce these parameters to have flexibility in picking $s$ and $t_l$,
which we will need later when we construct a simultaneous coreset
that uses a more elaborate set of initial centers $C$.

\begin{lemma}
  \label{lem:1d}
  Suppose $k \in \mathbb{Z}_+$, $\eps \in (0, 1)$, $X \subset \mathbb{R}^d$ is a data set, and $\mathcal{L}$ is a collection of lines in $\mathbb{R}^d$. Furthermore,
  \begin{itemize}
  	\item $X$ is partitioned into $\{X_l \mid l \in \calL\}$, where $X_l \subseteq l$ for $l \in \calL$, and
  	\item for each $l \in \calL$, $X_l$ is partitioned into a set of disjoint sub-intervals $\calY_l$, such that for each $Y \in \calY_l$, either $\mathrm{len}(I(Y)) \leq O(\frac{\eps}{p} \cdot s)$ or $\delta(Y) \leq O(\frac{\eps}{k} \cdot t_l )$ for some $s, t_l > 0$.
  \end{itemize}
  Then for all sets $C\subset \mathbb{R}^d$ of $k$ centers, the weighted set $D := \{ \mu(Y) \mid Y \in \calY_l, l \in \mathcal{L}  \}$ with weight $|Y|$ for element $\mu(Y)$, satisfies $|\cost_p(D, C) - \cost_p(X, C)| \leq O(\eps) \cdot (s + \sum_{l \in \calL}{t_l})$.
\end{lemma}
\ifeightpage
\begin{proof}
	The proof can be found in the full version.
\end{proof}
\else
\begin{proof}
Suppose $X = \{x_1, \ldots, x_n\}$.
The proof idea is similar to the $d=k=1$ case as in Section~\ref{sec:toy}. In particular, we construct an auxiliary set of points $Z:=\{ z_1, \ldots, z_n \}$, and the error bound is implied by bounding both $|\cost_p(Z, C) - \cost_p(D, C)|$ and $| \cost_p(Z, C) - \cost_p(X, C)|$ for all $k$-subset $C \subset \mathbb{R}^d$.

\paragraph{Notations}
For $x_i\in X$, let $Y_i \in \calY_l$ denote the unique sub-interval that contains $x_i$, where $l$ is the line that $Y_i$ belongs to, and let $\pi(x_i)$ denote the unique coreset point in $Y_i$ (which is $\mu(Y_i)$).
Define $M:=\{x_i \mid \mathrm{len}(Y_i)\leq O(\frac{\eps}{p} \cdot s)\}$ and $N := X\setminus M$.
We analyze the error for any given $C=\{c_1,\ldots, c_k\}$ and let $\calC_i=\{x\in X: \arg\min_{j\in [k]} d(x,c_j)=c_i\}$ (ties are broken arbitrarily) be the cluster induced by $C$. If $x\in \calC_i$, we say $x$ is served by $c_i$. Let $P_1\subset X$ denote the set of $p$ farthest points to $C$. Define $c_{il}$ to be the projection of $c_i$ onto line $l$.

\paragraph{Defining $Z$}
We define $z_i$ to be either $x_i$ or $\pi(x_i)$ as follows.
For $x_i\in M$, let $z_i=x_i$. For $x_i\in N$, if
\begin{itemize}
	\item[a)] $I(Y_i)$ does not contain any $c_{jl}$ for $j \in [n], l \in \calL$, and
	\item[b)] all points in $Y_i$ are served by a unique center, and
	\item[c)] $Y_i$ is either contained in $P_1$ or does not intersect $P_1$,
\end{itemize}
then we define $z_i := \pi(x_i)$ otherwise $z_i := x_i$.

Let $\mathrm{error}:=\eps \cdot (s + \sum_{l \in \calL}{t_l})$.
It suffices to show
$|\cost_p(D,C)-\cost_p(Z,C)|\leq O(\mathrm{error})$ and $|\cost_p(Z,C)-\cost_p(X,C)|\leq O(\mathrm{error})$.
\paragraph{Part I: $|\cost_p(D,C)-\cost_p(Z,C)|\leq O(\mathrm{error})$}
To prove $|\mathrm{cost}_p(D,C)-\cost_p(Z,C)|\leq O(\mathrm{error})$, we apply Lemma~\ref{lonelinf} with $S=M$, so
\begin{align}
& \quad |\mathrm{cost}_p(D,C)-\mathrm{cost}_p(Z,C)|\nonumber\\
&\leq p\cdot \max_{x_i\in M}|d(\pi(x_i),C)-d(z_i,C)|\nonumber\\
&\quad \quad+\sum_{x_i\in N} |d(\pi(x_i),C)-d(z_i,C)|.\label{eqn:1d_lm32}
\end{align}
Observe that $x_i \in M$ implies $|d(\pi(x_i), x_i)| \leq \mathrm{len}(I(Y_i)) = O(\frac{\eps}{p}\cdot s)$.
Therefore,
\begin{align}
	p\cdot \max_{x_i \in M}{| d(\pi(x_i), C) - d(z_i, C) |} \leq p \cdot O(\frac{\eps}{p}\cdot s) \cdot p = O(\eps \cdot s). \label{eqn:1d_m}
\end{align}

Then we bound the second term $\sum_{x_i \in N}{|d(\pi(x_i), C) - d(z_i, C)|}$.
We observe that
in each line $l$, there are only $O(k)$ distinct sub-intervals $Y_i \in \calY_l$ induced by $x_i \in N$ such that $z_i = x_i$.
Actually, for each line $l$, there are at most $k$ sub-intervals $Y \in \calY_l$ such that $I(Y)$ contains some $c_{il}$ for $i \in [n]$, and there are at most $2k$ sub-intervals whose points are served by at least $2$ centers, and there are at most $4k$ intervals that intersect $P_1$ but are not fully contained in $P_1$.
Hence,
\begin{align}
	\sum_{x_i \in N}{ |d(\pi(x_i), C) - d(z_i, C)| }
	&\leq \sum_{x_i \in N : z_i = x_i}{ d(\pi(x_i), x_i) } \nonumber \\
	&\leq \sum_{Y_i: x_i \in N , z_i = x_i}{\delta(Y_i)} \nonumber \\
	&\leq \sum_{l\in \calL}{O(k) \cdot \frac{\eps}{k}\cdot t_l}
	= O(\eps) \cdot \sum_{l \in \calL}{t_l}.\label{eqn:1d_n}
\end{align}
Combining Inequality~\ref{eqn:1d_m} and~\ref{eqn:1d_n} with~\ref{eqn:1d_lm32}, we conclude that $|\cost_p(D, C) - \cost_p(Z, C)| \leq O(\mathrm{error})$.

\paragraph{Part II: $|\cost_p(Z,C)-\cost_p(X,C)|\leq O(\mathrm{error})$}
Let $\calY' \subseteq \bigcup_{l \in \calL}{\calY_l}$ be the set of sub-intervals $Y_i$ such that $x_i \in N$ and a) - c) hold (i.e. $z_i = \pi(x_i)$).
Note that by construction, the only difference between $Z$ and $X$ is due to replacing points in sub-intervals $Y \in \calY'$ with $|Y|$ copies of $\mu(Y)$, thus it suffices to analyze this replacement error.

Let $P_Z \subseteq Z$ be (multi)-set of the $p$-furthest points of $Z$ from $C$.
We start with showing that, the coreset point $\mu(Y)$ of $Y \in \calY'$, is fully contained in $P_Z$ or does not intersect $P_Z$. Consider some $Y \in \calY'$ and assume points in $Y$ are all served by $c_j$.
Denote the endpoints of interval $I(Y)$ as $a$ and $b$.
Let $l \in \calL$ be the line that contains $Y$. Since $I(Y)$ does not contain $c_{jl}$, then either $\angle abc_j > \frac{\pi}{2}$ or $\angle bac_j > \frac{\pi}{2}$. W.l.o.g., we assume that $\angle abc_j>\frac{\pi}{2}$. By Observation~\ref{geobs}, we know that if $Y \cap P_1 = \emptyset$, then $\mu(Y) \not\in P_Z$; on the other hand, if $Y \subseteq P_1$, then $\mu(Y) 
\in P_Z$.
\begin{observation}\label{geobs}
	Let $\Delta ABC$ denote a triangle where $\angle ABC>\frac{\pi}{2}$. Let $E$ be a point on the edge $BC$ then $|AC|\geq |AE|\geq |AB|$.
\end{observation}
Hence, if $Y \in \calY'$ has empty intersection with $P_1$, the $|Y|$ copies of $\mu(Y)$ in $Z$ does not contribute to either $\cost_p(Z, C)$ or $\cost_p(X, C)$.
Thus, it remains to bound the error for $Y \in \calY'$ such that $Y \subseteq P_1$.
In the 1-dimensional case as in Section~\ref{sec:toy}, replacing $Y$ with the mean $\mu(Y)$ does not incur any error, as the center is at the same line with the interval $I(Y)$.
However, this replacement might incur error in the $d$-dimensional case since the center might be \emph{outside} the line that contains the sub-interval $Y$.
Luckily, this error has been analyzed in~\cite[Lemma 2.8]{har2007smaller}, and we adapt their argument in Lemma~\ref{var} (shown below). By Lemma~\ref{var}, $\forall c_j \in C, l \in \calL$, the total replacement error for all sub-intervals $Y \in \calY' \cap \calY_l$ such that i) $Y \subseteq P_1$ and ii) all points in $Y$ are served by $c_j$, is at most $ O(\frac{\eps}{k} \cdot t_l )$.
Therefore,
\begin{align*}
|\cost_p(Z,C)-\cost_p(X,C)|
&\leq \sum_{l\in \calL}\sum_{c_j \in C}O(\frac{\eps}{k} \cdot t_l ) \\
&= O(\eps)\cdot \sum_{l \in \calL}{t_l}
\leq O(\mathrm{error}).
\end{align*}
This finishes the proof of Lemma~\ref{lem:1d}.
\end{proof}
\fi

\begin{lemma}\label{var}
Let $l\subset \mathbb{R}^d$ be a line and $c\in \mathbb{R}^d$ be a point. Define $c_l$ be the projection of $c$ on $l$. Assume that $X_1,\ldots,X_m\subset l$ are finite sets of points in $l$ such that 
$I(X_1),\ldots,I(X_m)$ are disjoint, $\delta(X_i)\leq r$ and $\forall i \in [m]$, $c_l \notin I(X_i)$.
Then
$$
\left|\sum_{i=1}^m \sum_{x\in X_i} d(x,c)-\sum_{i=1}^m |X_i|\cdot d(\mu(X_i),c)\right|\leq O(r).
$$
\end{lemma}

\begin{proof}
W.l.o.g. we assume $\forall i \in [m]$, $X_i$ is not a singleton, since $\sum_{x \in X_i}{d(x, c)} = |X_i| \cdot d(\mu(X_i), c)$ for singleton $X_i$.
Let $\mathrm{err}_i := \sum_{x \in X_i}{d(x, c)} - |X_i| \cdot d(\mu(X_i), c)$. In~\cite[Lemma 2.5]{har2007smaller}, it was shown that $\mathrm{err}_i \geq 0$ for all $i \in [m]$ (using that $c_l \not\in I(X_i)$).
Furthermore, it follows from the argument of~\cite[Lemma 2.8]{har2007smaller} that, if each $X_i$ is modified into a weighted set $X_i'$
with \emph{real} weight $w_i : X_i' \to \R_+$,
such that $I(X_i) = I(X_i')$ and each $X_i'$ has the same cumulative error $\delta_{w_i}(X_i') = r$, then
\begin{itemize}
	\item $\forall i \in [m]$, $\mathrm{err}_i' :=  \sum_{x \in X_i'}{w(x)\cdot d(x, c)} - \left(\sum_{x\in X_i'}{w_i(x)}\right) \cdot d(\mu_{w_i}(X_i'), c) \geq 0$, and
	\item 
		$\sum_{i \in [m]}{\mathrm{err}_i'} = \sum_{i \in [m]}\sum_{x\in X_i'} w(x) \cdot d(x,c)-\sum_{i \in [m]}\left(\left(\sum_{x \in X_i'}w_i(x)\right)\cdot d(\mu_{w_i}(X_i'),c)\right) \leq O(r)$.
\end{itemize}
Hence, it suffices to show that it is possible to modify each $X_i$ into a real weighted set $X_i'$ with $I(X_i') = I(X_i)$, such that $\delta(X_i') = r$ and $\mathrm{err}_i' \geq \mathrm{err}_i$ for all $i\in [m]$.

For each $i \in [m]$, find two points $a\neq b\in I(X_i)$ such that $\mu(X_i)$ is the midpoint of $a$ and $b$. Such $a$ and $b$ must exist since we assume $X_i$ is not a singleton.
We form $X_i'$ by adding points $a$ and $b$ with the same (real-valued) weight into $X_i$, such that $\delta_{w_i}(X_i')=r$,
then $\mathrm{err}_i' \geq \mathrm{err}_i$ follows from the geometric fact that $d(c, a)+d(c, b)\geq 2 d(c, \mu(X_i))$.
This concludes Lemma~\ref{var}.
\end{proof}

 \subsection{Coreset for \pCentrum in $\mathbb{R}^d$}
\label{sec:highdim}
We now prove the theorem about a coreset for \pCentrum.
As discussed above, we use a projection procedure inspired by~\mycite{har2007smaller} to reduce to line cases,
and then apply Lemma~\ref{lem:1d} to get the coreset.

\begin{theorem}
  \label{thm:pcent}
  Given $k\in \mathbb{Z}_+$, $\eps \in (0,1)$, an $n$-point data set $X\subset \mathbb{R}^d$, and $p \in [n]$, there exists an $\eps$-coreset $D \subset \mathbb{R}^d$ of size $O(\frac{k^2}{\eps^{d+1}})$ for $\pCentrum$.
  Moreover, it can be computed in polynomial time.
\end{theorem}
We start with a detailed description of how we reduce to the line case.
This procedure will be used again in the simultaneous coreset.

\paragraph{Reducing to Lines: Projection Procedure}
Consider an $m$-point set $C := \{c_1, \ldots, c_m\} \subset \mathbb{R}$ which we call projection centers. We will define a new data set $X'$ by projecting points in $X$ to some lines defined with respect to $C$. The lines are defined as follows. For each $c_i \in C$, construct an $\eps$-net $N_i$ for the unit sphere centered at $c_i$, and for $u \in N_i$, define $l_{i u}$ as the line that passes through $c_i$ and $u$.
Let $\mathcal{L} := \{ l_{iu} \mid i \in [m], u \in N_i \}$ be the set of projection lines. Then $X'$ is defined by projecting each data point $x \in X$ to the nearest line in $\mathcal{L}$. Since $N_i$'s are $\eps$-nets on unit spheres in $\mathbb{R}^d$, we have $|\mathcal{L}| \leq O(\frac{1}{\eps})^d \cdot |C|$. The cost of this projection is analyzed below in Lemma~\ref{lem:project_cost}.
\begin{lemma}[projection cost]
	\label{lem:project_cost}
	For all $C' \subset \mathbb{R}^d$ and $p \in [n]$, $|\cost_p(X', C') - \cost_p(X, C')| \leq O(\eps) \cdot \cost_p(X, C)$.
\end{lemma}
\ifeightpage
\begin{proof}
	The proof can be found in the full version.
\end{proof}
\else
\begin{proof}
	For $x \in X$, denote the projection of $x$ by $\sigma(x) \in X'$, and for $y \in X'$ let $\sigma^{-1}(y) \in X$ be any $x \in X$ such that $\sigma(x) = y$.
	Observe that $\forall x \in X$, $d(x, C') - d(\sigma(x), C') \leq d(x, \sigma(x)) \leq O(\eps) \cdot d(x, C)$, where the first inequality is by triangle inequality, and the last inequality is by the definition of $\sigma(x)$.
	
	Let $A \subseteq X'$ be the $p$-furthest points from $C'$ in $X'$.
	Then
	\begin{align*}
	\cost_p(X', C')
	&= \sum_{x \in A}{d(x, C')}
	\leq \sum_{x \in A}{d(\sigma^{-1}(x), C')} + O(\eps) \cdot \sum_{x \in A}{d(\sigma^{-1}(x), C)} \\
	&\leq \cost_p(X, C') + O(\eps) \cdot \cost_p(X, C). 
	\end{align*}
	Similarly, let $B \subseteq X$ be the $p$-furthest points from $C'$ in $X$.
	Then
	\begin{align*}
		\cost_p(X, C')
		&= \sum_{x \in B}{d(x, C')}
		\leq \sum_{x \in B}{ d(\sigma(x), C') + O(\eps) \cdot \sum_{x \in B}{ d(x, C) } } \\
		&\leq \cost_p(X', C') + O(\eps) \cdot \cost_p(X, C).
	\end{align*}
	This finishes the proof.
\end{proof}
\fi

We remark that both the projection center and the candidate center $C'$ in Lemma~\ref{lem:project_cost} are not necessarily $k$-subsets. This property is not useful for the coreset for \pCentrum, but it is crucially used in the simultaneous coreset in Section~\ref{sec:order}.
\ifeightpage
The remaining details of the proof for Theorem~\ref{thm:pcent} can be found in the full version.
\else
Below we give the proof of Theorem~\ref{thm:pcent}.
\begin{proofof}{Theorem~\ref{thm:pcent}}
  For the purpose of Theorem~\ref{thm:pcent}, we pick $C$ to be an $O(1)$-approximation to the optimal centers for the $k$-facility \pCentrum on $X$, i.e., $\cost_p(X, C) \leq O(1) \cdot \opt_p$, where $\opt_p$ is the optimal value for the \pCentrum.
  Such $C$ may be found in polynomial time by applying known approximation algorithms, say by~\mycite{chakrabarty2018approximation}.
As analyzed in Lemma~\ref{lem:project_cost}, for such choice of $C$, the error incurred in $X'$ because of the projections is bounded by $O(\eps) \cdot \opt_p$.

We will apply Lemma~\ref{lem:1d} on $X'$.
The line partitioning $\{ X'_l \mid l \in \calL \}$ of $X'$ that we use in Lemma~\ref{lem:1d} is naturally induced by the line set $\calL$ resulted from the projection procedure.
Then, for each $l \in \calL$, we define the disjoint sub-intervals $\calY_l$ as follows.
Let $S := X' \cap l$, let $S_1 \subseteq S$ be the subset of the $p$-furthest points from $C$, and let $S_2 := S \setminus S_1$. We then break $S_1$ and $S_2$ into sub-intervals, using similar method as in Section~\ref{sec:toy}. Let $\apx := \cost_p(X', C)$, and let $\apx_{l}$ be the contribution of $S$ in $\apx$. Break $S_1$ into sub-intervals according to cumulative error $\delta$ with threshold $O(\frac{\eps \cdot \apx_{l}}{k})$, similar with how we deal with $L$ and $R$ in Section~\ref{sec:toy}. Break $S_2$ into maximal sub-intervals of length $\Theta(\frac{\eps \cdot \apx}{p})$, similar with $Q$ in Section~\ref{sec:toy}. Again, similar with the analysis in Section~\ref{sec:toy}, the number of sub-intervals is at most $O(\frac{k}{\eps})$ for each $l \in \calL$.

Finally, we apply Lemma~\ref{lem:1d} with $t_l := \apx_{l}$ and $s := \apx$, and it yields a multi-set $D$ such that
$|\cost_p(D, C') - \cost_p(X, C')| \leq O(\eps) \cdot (\apx + \sum_{l \in \calL}{\apx_l}) = O(\eps) \cdot \apx = O(\eps) \cdot \opt_p$.
On the other hand, the size of $D$ is upper bounded by $|\calL| \cdot \frac{k}{\eps} \leq O(\frac{k}{\eps^{d+1}}) \cdot |C| \leq O(\frac{k^2}{\eps^{d+1}}) $. This concludes Theorem~\ref{thm:pcent}.
\end{proofof}
\fi

 \subsection{Simultaneous Coreset for \OkM in $\mathbb{R}^d$}
\label{sec:order}
In this section we prove our main theorem that is stated below as Theorem~\ref{thm:okm}. As discussed before, we first show it suffices to give simultaneous coreset for \pCentrum for $O(\log n)$ values of $p$.
Then we show how to combine these coresets to obtain a simultaneous coreset.

\begin{theorem}
  \label{thm:okm}
  Given $k\in \mathbb{Z}_+$, $\eps\in (0,1)$ and an $n$-point data set $X\subset \mathbb{R}^d$,
  there exists a simultaneous $\eps$-coreset of size $O(\frac{k^2\log^2 n}{\eps^d})$ for \OkM. Moreover, it can be computed in polynomial time.
\end{theorem}

We start with the following lemma, which reduces simultaneous coresets for \OkM to simultaneous coresets for \pCentrum.

\begin{lemma}\label{centogen}
Suppose $k\in \mathbb{Z}_+$, $\varepsilon\in (0,1)$, $X\subset \mathbb{R}^d$ and $D$ is a simultaneous $\eps$-coreset for the $k$-facility \pCentrum problem for all $p\in [n]$. Then $D$ is a simultaneous $\eps$-coreset for \OkM.
\end{lemma}

\ifeightpage
\begin{proof}
	The proof can be found in the full version.
\end{proof}
\else
\begin{proof}
Suppose $X = \{x_1, \ldots, x_n\}$.
We need to show for any center $C$ and any weight $v$, $\cost_v(D, C) \in (1\pm \eps) \cdot \cost_v(X, C)$.
Fix a center $C$ and some weight $v$. We assume w.l.o.g. $d(x_1,C)\geq \ldots\geq d(x_n,C)$. By definition we have $\cost_v(X,C)=\sum_{i=1}^n v_i\cdot d(x_i,C)$ and $\cost_p(X,C)=\sum_{i=1}^p d(x_i,C)$ for any $p$.
Since $D$ is an $\varepsilon$-coreset of $X$ for \pCentrum on every $p\in [n]$, $\cost_p(D,C)\in (1\pm\varepsilon) \cost_p(X,C)$. Let $v_{n+1} := 0$, and we have
\begin{align*}
\cost_v(D,C)
&=\sum_{p=1}^n (v_p-v_{p+1})\cdot \cost_p(D,C) \\
&\in (1\pm\eps) \cdot \sum_{p=1}^n (v_p-v_{p+1})\cdot \cost_p(X,C) \\
&=(1\pm\eps) \cdot \cost_v(X,C).
\end{align*}
\end{proof}
\fi

With the help of the following lemma, we only need to preserve the objective for $p$'s taking powers of $(1+\eps)$. In other words, it suffices to construct simultaneous coresets to preserve the objective for only $O(\frac{\log n}{\eps})$ distinct values of $p$'s.

\begin{lemma} \label{logn}
Let $X,C\subset \mathbb{R}^d$ and $p_1,p_2\in[n]$ such that $p_1\leq p_2\leq (1+\eps) \cdot p_1$. Then $$
\cost_{p_1}(X,C)\leq \cost_{p_2}(X,C)\leq (1+\eps)\cdot \cost_{p_1}(X,C).
$$
\end{lemma}
\ifeightpage
\begin{proof}
	The proof can be found in the full version.
\end{proof}
The remaining details of the proof for Theorem~\ref{thm:okm} can be found in the full version.
\else
\begin{proof}
We assume w.l.o.g. $d(x_1,C)\geq \ldots \geq d(x_n,C)$.
By definition,
\begin{align*}
	\cost_{p_2}(X, C) = \cost_{p_1}(X, C) + \sum_{i = p_1}^{p_2}{d(x_i, C)} \geq \cost_{p_1}(X, C).
\end{align*}
On the other hand,
\begin{align*}
\cost_{p_2}(X,C)
&=\cost_{p_1}(X,C)+\sum_{i=p_1+1}^{p_2} d(x_i,C)\\
&\leq \cost_{p_1}(X,C)+(p_2-p_1)\cdot \frac{1}{p_1}\cost_{p_1}(X,C)\\
&\leq \cost_{p_1}(X,C)+\eps\cdot p_1\cdot \frac{1}{p_1}\cdot \cost_{p_1}(X,C)\\
&=(1+\eps)\cdot \cost_{p_1}(X,C).
\end{align*}
\end{proof}
We are now ready to present the proof of Theorem~\ref{thm:okm}.
\begin{proofof}{Theorem~\ref{thm:okm}}
	As mentioned above, by Lemma~\ref{logn}, it suffices to obtain an $\eps$-coreset for $O(\frac{\log n}{\eps})$ values of $p$'s. Denote the set of these values of $p$'s as $W$.
	
	We use a similar framework as in Theorem~\ref{thm:pcent}, and we start with a projection procedure. However, the projection centers are different from those in Theorem~\ref{thm:pcent}. For each $p \in W$, we compute an $O(1)$-approximate solution $C_p$ for \pCentrum, which is a $k$-subset. Then, we define $C := \bigcup_{p \in P}{C_p}$ be the union of all these centers, so $|C| \leq O(\frac{k \cdot \log n}{\eps})$.
	Let $X'$ be the projected data set.
	By Lemma~\ref{lem:project_cost}, the projection cost is bounded by $O(\eps) \cdot \cost_p(X, C)  \leq O(\eps) \cdot \cost_p(X, C_p) \leq O(\eps) \cdot \opt_p$, for all $p \in W$.
	
	Following the proof of Theorem~\ref{thm:pcent}, for each $p \in W$, we apply Lemma~\ref{lem:1d} on the projected set $X'$ in exactly the same way, and denote the resulted coreset as $D_p$. By a similar analysis, for each $p$, the size of $D_p$ is $O(\frac{{k^2 \cdot \log n}}{\eps^{d+1}})$.
	
	Then we describe how to combine $D_p$'s to obtain the simultaneous coreset. A crucial observation is that, the coresets $D_p$'s are constructed by replacing sub-intervals with their mean points, and for all $p \in W$, the $D_p$'s are built on the same set of lines.
	Therefore, we can combine the sub-intervals resulted from all $D_p$'s.
	Specifically, combining two intervals $[a, b]$ and $[c, d]$ yields $[\min\{a, c\}, \max\{ a, c \}]$, $[ \max\{a, c\}, \min\{b, d\} ]$, $[\min\{b, d\}, \max\{b, d\}]$.
	For any particular $p$, in the combined sub-intervals, the length upper bound and the $\delta$ upper bound required in Lemma~\ref{lem:1d} still hold. Hence the coreset $D$ resulted from the combined sub-intervals is a simultaneous coreset for all $p \in W$. By Lemma~\ref{centogen} and Lemma~\ref{logn}, $D$ is a simultaneous $\eps$-coreset for \OkM.
	
	The size of $D$ is thus $O(\log n)$ times the coreset for a single $p$. Therefore, we conclude that the above construction gives a simultaneous $\eps$-coreset with size $O(\frac{k^2 \log^2 n}{\eps^{d+1}})$,
        which completes the proof of Theorem~\ref{thm:okm}.
\end{proofof}
\fi

 \subsection{Lower Bound for Simultaneous Coresets}
\label{sec:LB}

In this section we show that the size of a simultaneous coreset for $\OkM$,
and in fact even for the special case \pCentrum,
must grow with $n$, even for $k = d = 1$. 
More precisely, we show that it must depend at least logarithmically on $n$,
and therefore our upper bound in Theorem~\ref{thm:okm}
is nearly tight with respect to $n$.

\begin{theorem}
  \label{thm:lb}
  For every (sufficiently large) integer $n$ and every $n^{-1/3} < \epsilon < 1/2$,
  there exists an $n$-point set $X \subset \mathbb{R}$, such that any simultaneous $\epsilon$-coreset of $X$ for \pCentrum with $k=1$ has size $\Omega(\eps^{-1/2}\log n)$.
\end{theorem}

While a simultaneous coreset preserves the objective value for all possible centers (in addition to all $p\in[n]$), 
our proof shows that even \emph{one specific} center
already requires $\Omega(\log n)$ size.
Our proof strategy is as follows.
Suppose $D$ is a simultaneous $\eps$-coreset for \OkM on $X \subset \mathbb{R}$ with $k = 1$, and let $c \in \mathbb{R}$ be some center to be picked later.
Since $D$ is a simultaneous coreset for \OkM, it is in particular a coreset for \pCentrum problems for all $p \in [n]$.
Let $W_X(p) := \cost_p(X, c)$ be the cost as a function of $p$, 
and let $W_D(p)$ be similarly for the coreset $D$,
when we view $X$, $D$ and the center $c$ as fixed. 
It is easy to verify that $W_D(\cdot)$ is a piece-wise linear function with only $O(|D|)$ pieces.
Now since $D$ is a simultaneous $\epsilon$-coreset,
the function $W_D(\cdot)$ has to approximate $W_X(\cdot)$ in the entire range,
and it would suffice to find an instance $X$ and a center $c$
for which $W_X(p)$ cannot be approximated well by a few linear pieces.
(Note that this argument never examines the coreset $D$ explicitly.)
The detailed proof follows. 

\begin{proof}
  Throughout, let $F(x) := \sqrt{x}$.
  Now consider the point set $X := \{x_1, x_2, \ldots, x_n\}\subset\R$,
  defined by its prefix-sums $\sum_{j\in [i]}{x_j} = F(i)$ (for all $i \in [n]$).
  It is easy to see that $1 = x_1 > x_2 > \cdots > x_n > 0$.
  Fix center $c := 0$ and consider a simultaneous $\eps$-coreset $D$ of size $|D|$. 
  Since $D$ is a simultaneous coreset for \OkM, $W_D(p) \in (1\pm \eps) \cdot W_X(p)$ for all $p \in [n]$,
  where by definition $W_X(p) = F(p)$.
	
  We will need the following claim,   which shows that each linear piece in $W_D(\cdot)$ (denote here by $g$) 
  cannot be too ``long'', as otherwise the relative error exceeds $\eps$. 
  We shall use the notation $[a..b]=\{a,a+1,\ldots,b\}$ for two integers $a\le b$.

  \begin{claim}
    \label{claim:global_lb_gap}
    Let $F$ be as above and let $g:\R\to\R$ be a linear function. 
    Then for every two integers $a\ge 1$ and $b \geq (1 + \frac{1}{\eps})^2$
    satisfying $\frac{b}{a} \geq (1+12\sqrt{\eps})^4$,
    there exists an integer $p \in [a..b]$ such that
    $g(p) \notin (1\pm \eps) \cdot F(p)$. 
  \end{claim}

  \begin{proof}
    We may assume both $g(a) \in (1\pm \eps)\cdot F(a)$ and $g(b) \in (1\pm \eps)\cdot F(b)$, as otherwise the claim is already proved.
		Since $g$ is linear, it is given by $g(x) = k (x - a) + g(a)$, where $k := \frac{g(b) - g(a)}{b-a}$.
		Let $\widehat{p} := \lfloor \sqrt{ab} \rfloor$. Observe that $\widehat{p} \in [a..b]$,
		thus it suffices to prove that 
		\begin{align*}
                  \frac{g(\widehat{p})}{F(\widehat{p})}
                  < 1 - \eps.
		\end{align*}
                
The intuition of picking $\widehat{p} = \lfloor \sqrt{ab} \rfloor$
is that $x = \sqrt{ab}$ maximizes $\frac{F(x)}{\widehat{g}(x)}$,
where $\widehat{g}(x)$ is the linear function passing through $(a, F(a))$
and $(b, F(b))$,
and we know that this $\widehat{g}(x)$ should be ``close'' to $g$ as $g(a) \in (1\pm \eps) \cdot F(a)$ and $g(b) \in (1\pm \eps) \cdot F(b)$.
(Notice that since $F$ is concave,
$\widehat{g}(x) \leq F(x)$ for all $x \in [a, b]$.)

To analyze this more formally, 
		\begin{align*}
		g(\widehat{p})
		&=\frac{g(b)-g(a)}{b-a} (\lfloor \sqrt{ab} \rfloor - a) + g(a) \\
		&=\frac{ (\lfloor \sqrt{ab} \rfloor - a) \cdot g(b) + (b - \lfloor \sqrt{ab}\rfloor) \cdot g(a) }{b-a} \\
		&\in (1\pm \epsilon) \cdot \frac{ (\lfloor \sqrt{ab} \rfloor - a)\sqrt{b} + (b - \lfloor \sqrt{ab}\rfloor)\sqrt{a} }{b-a} \\
		& = (1\pm \epsilon) \cdot  \frac{\lfloor \sqrt{ab} \rfloor + \sqrt{ab}}{\sqrt{a} + \sqrt{b}}.
		\end{align*}
    Therefore,
		\begin{align*}
		0 < \frac{g(\widehat{p})}{F(\widehat{p})}
		\leq (1 + \epsilon) \cdot \frac{ \lfloor \sqrt{ab} \rfloor + \sqrt{ab} }{ (\sqrt{a} + \sqrt{b})\cdot \sqrt{\lfloor \sqrt{ab} \rfloor} }.
                \end{align*}
To simplify notation,
let $t := \left(\frac{b}{a}\right)^{1/4}$ and $s := \left(\frac{\lfloor \sqrt{ab} \rfloor}{\sqrt{ab}} \right)^{1/2}$.
By simple calculations, our assumptions
$\frac{b}{a} \geq (1 + 12\sqrt{\eps})^4$ and $b \geq (1 + \frac{1}{\eps})^2$
imply the following facts: 
$t + t^{-1} \geq 2 + 11\eps$ and $\frac{1}{1+\eps} \leq s \leq 1$.
And now we have
\begin{align*}
	(1+\eps) \cdot \frac{ \lfloor \sqrt{ab} \rfloor + \sqrt{ab} }{ (\sqrt{a} + \sqrt{b})\cdot \sqrt{\lfloor \sqrt{ab} \rfloor} }
	&= (1+\eps)\cdot \frac{ s + s^{-1} }{ t + t^{-1} } \nonumber \\
	&\leq \frac{(1+\eps)(2+\eps)}{2 + 11\eps}  \nonumber \\
	&< 1 - \eps . \end{align*}
Altogether,
we obtain
$0 < \frac{g(\widehat p)}{F(\widehat p)} < 1 - \eps$, 
which completes the proof of Claim~\ref{claim:global_lb_gap}.
\end{proof}

We proceed with the proof of Theorem~\ref{thm:lb}.
Recall that $W_D(p)$ is the sum of the $p$ largest distances from points
in $D$ to $c$, when multiplicities are taken into account.
Thus, if the $p$-th largest distance for all $p\in[a..b]$
arise from the same point of $D$ (with appropriate multiplicity), 
then $W_D(p)-W_D(p-1)$ is just that distance, regardless of $p$,
which means that $W_D(p)$ is linear in this range. 
It follows that $W_D(p)$ is piece-wise linear with at most $|D|$ pieces. 
By Claim~\ref{claim:global_lb_gap} and the error bound of the coreset,
if a linear piece of $W_D(p)$ spans $p=[a..b]$ where $(1 + \frac{1}{\eps})^2 \leq b  \leq n$, then $b \leq (1 + 12 \sqrt{\eps})^4 \cdot a \leq (1 + O(\sqrt{\eps})) \cdot a$.
Since all the linear pieces span together all of $[n]$, we conclude that
$|D| = \Omega\left(\log_{1 + O(\sqrt{\eps})}(\eps^2 n) \right) = \Omega\left(\frac{\log n}{\sqrt{\eps}}\right)$
and proves Theorem~\ref{thm:lb}.
\end{proof}

 \section{Experiments}
\label{sec:exp}

We evaluate our coreset algorithm experimentally on real 2D geographical data. Our data set is the whole Hong Kong region extracted from OpenStreetMap~\cite{OpenStreetMap}, with complex objects such as roads replaced with their geometric means. The data set consists of about 1.5 million 2D points and is illustrated in Figure~\ref{fig:data_set}.
Thus, $d=2$ and $n\approx 1.5\cdot 10^6$ throughout our experiments.

\begin{figure*}[ht]
	\centering
	\captionsetup{font=small}
	\begin{subfigure}[b]{0.48\linewidth}
		\centering
		\includegraphics[width=\linewidth]{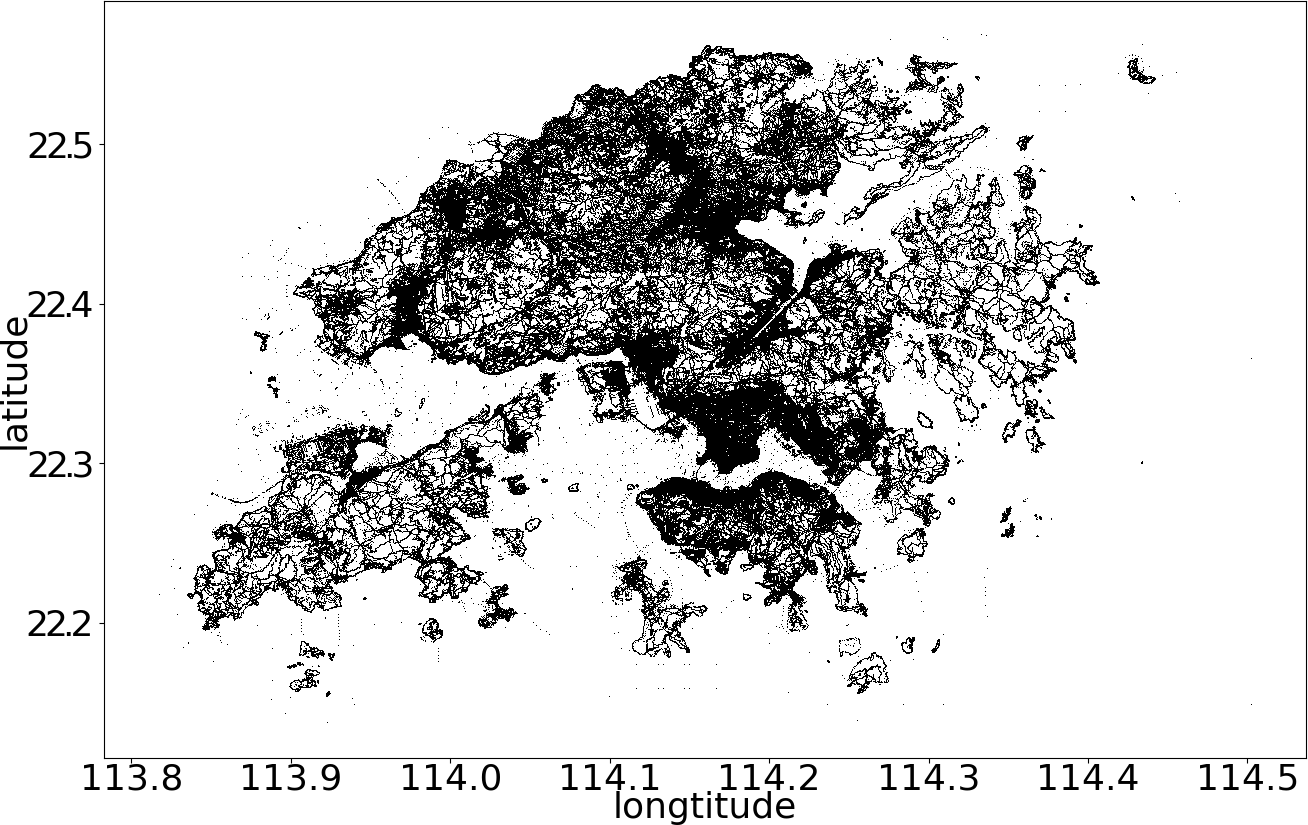}
	\end{subfigure}
        \qquad
	\begin{subfigure}[b]{0.447\linewidth}
		\centering
		\includegraphics[width=\linewidth]{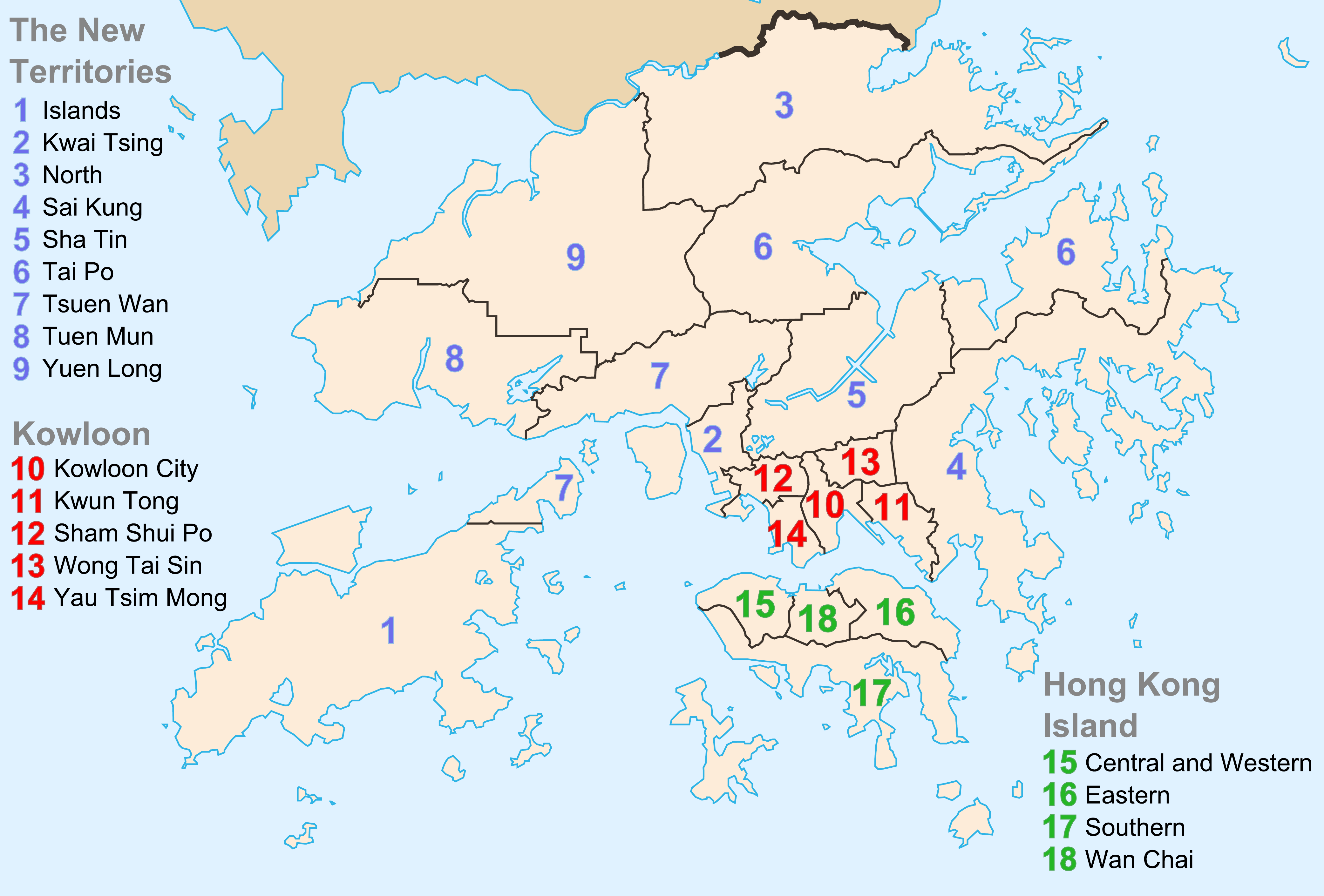}
	\end{subfigure}
  \caption{Demonstration of the data set.
    The 2D points extracted from~\cite{OpenStreetMap} are plotted on the left, 
    next to a map of Hong Kong~\cite{wiki:hk} on the right. 
  }
  \label{fig:data_set}
  \hrulefill 
\end{figure*}

\paragraph{Implementation}
Recall that our coreset construction requires an initial center set $C$
that is an $O(1)$-approximation for the \pCentrum problem.
However, \pCentrum is NP-hard as it includes \kCenter
(which is NP-hard even for points in $\R^2$),
and polynomial-time $O(1)$-approximation algorithms known for it~\cite{byrka2018constant,chakrabarty2017interpolating}
are either not efficient enough for our large data set or too complicated to implement.
Our experiments deal with an easier problem (small $k$ and points in $\R^2$),
but since we are not aware of a good algorithm for it,
our implementation employs instead the following simple heuristic:
sample random centers from the data points multiple times, 
and take the sample with the best (smallest) objective value.

Our first experiment evaluates the performance of this heuristic.
The results in Figure~\ref{fig:sampling_heuristic} show that
30 samples suffice to obtain a good solution for our data set.
The rest the algorithm is implemented following the description
in Section~\ref{sec:simul},
while relying on the above heuristic as if it achieves $O(1)$-approximation. 
Thus, the experiments in this section for various $\eps$, $p$ and $k$,
all evaluate a version of the algorithm that uses the heuristic.

\begin{figure*}[ht]
	\centering
	\captionsetup{font=small}
	\begin{subfigure}[b]{0.49\linewidth}
		\includegraphics[width=\linewidth]{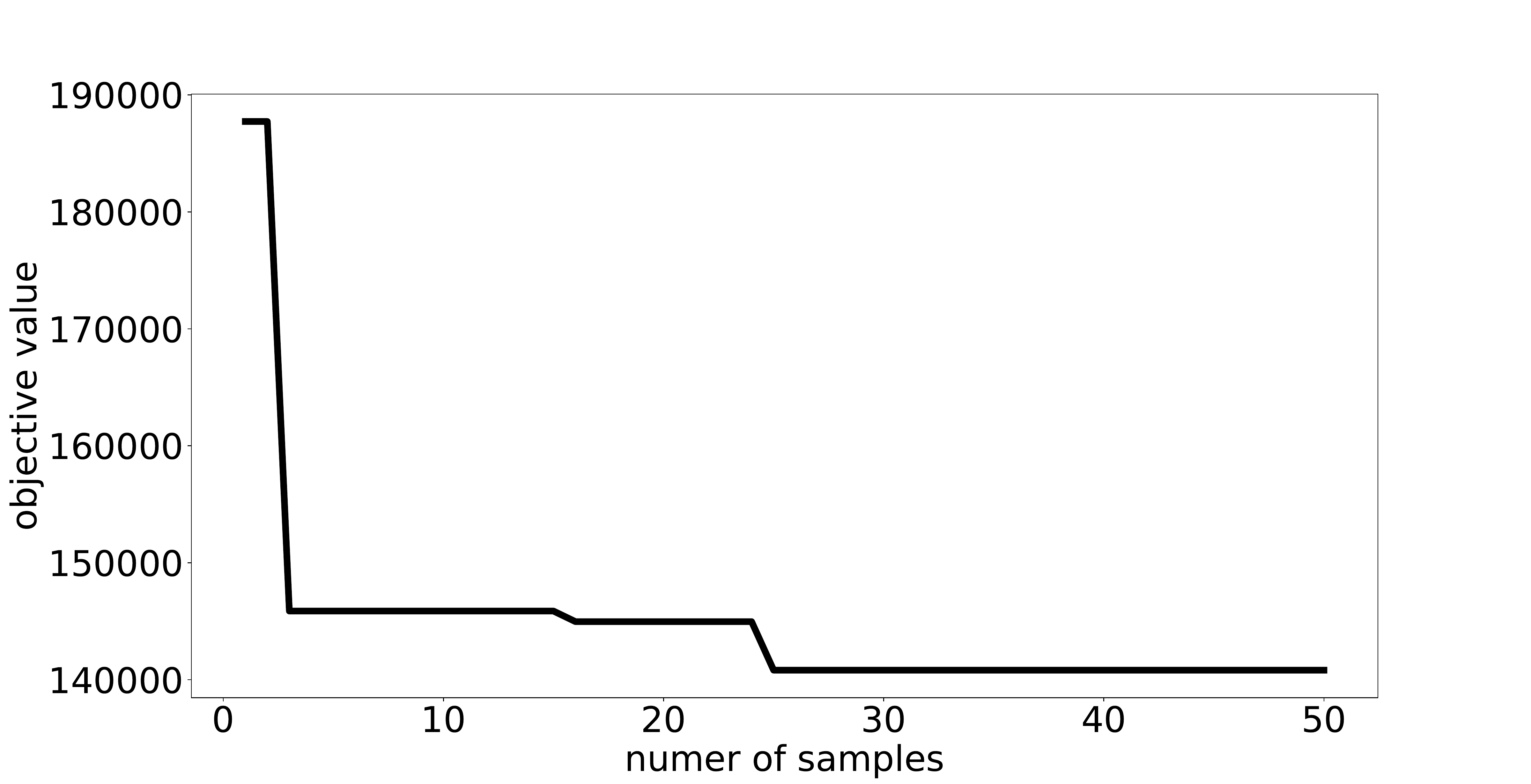}
		\caption{$p = n$, $k = 2$}
	\end{subfigure}
	\begin{subfigure}[b]{0.49\linewidth}
		\includegraphics[width=\linewidth]{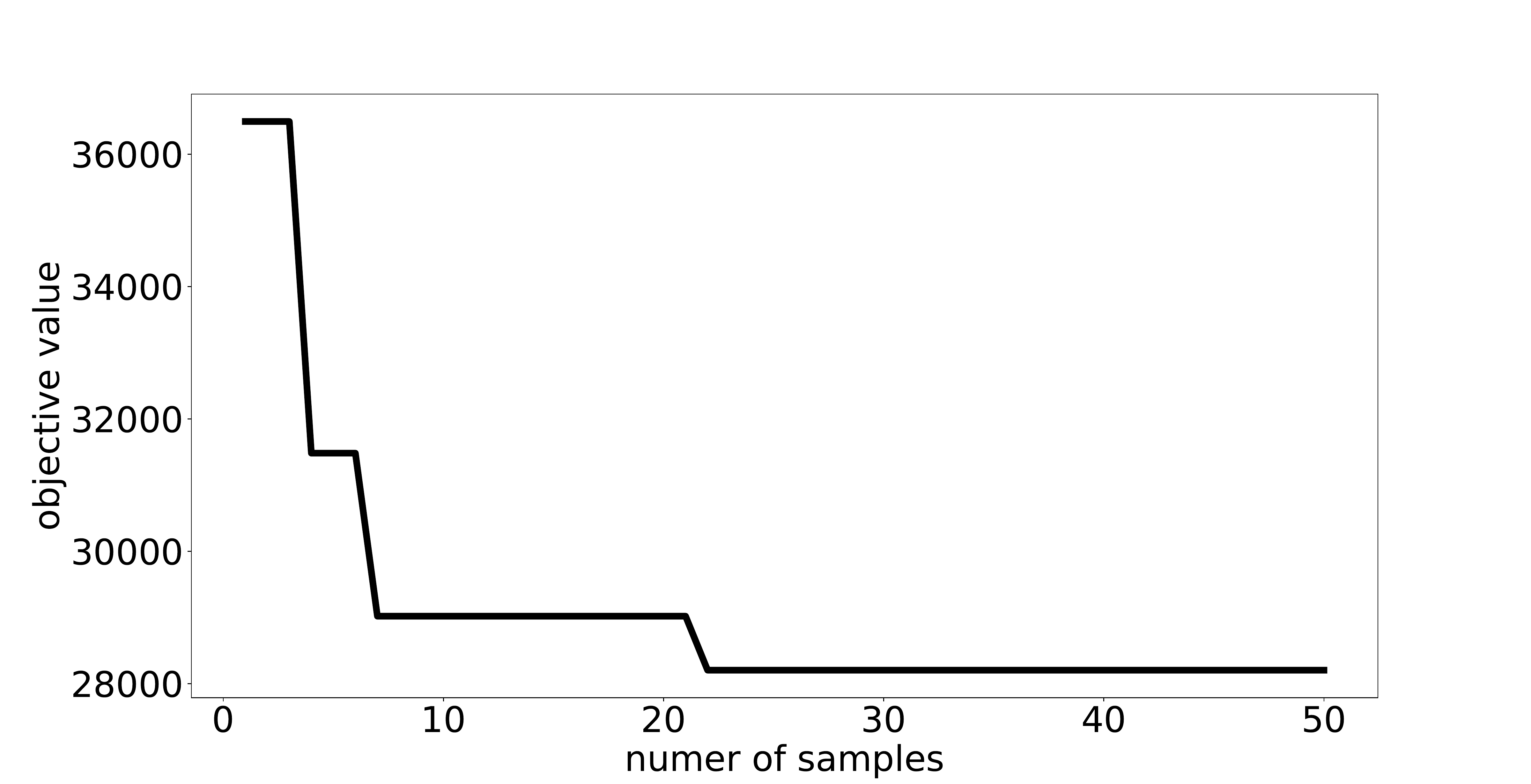}
		\caption{$p = 0.1 n$, $k = 2$}
	\end{subfigure}
	\caption{Performance of our \pCentrum heuristic, which takes the best of multiple randomly sampled centers. }
	\label{fig:sampling_heuristic}
  \hrulefill 
\end{figure*}

\paragraph{Performance Evaluation}
To examine the performance of our coreset algorithm for \pCentrum 
(using the heuristic for the initial centers),
we execute it with parameters $p = 0.1n$ and $k = 2$,
and let the error guarantee $\epsilon$ vary, 
to see how it affects the empirical size and error of the coreset. 
To evaluate the empirical error, we sample 100 random centers (each consisting of $k=2$ points) from inside the bounding box of the data set,
and take the maximum relative error, 
where the relative error of coreset $X'$ on centers $C$ is defined as $\left|\frac{\mathrm{cost}(X', C)}{\mathrm{cost}(X, C)} - 1\right|$
(similarly to how we measure $\epsilon$).
We report also the total running time for computing the objective for the above mentioned 100 random centers, comparing between the original data set $X$ and on the coreset $X$', denoted by $T_X$ and $T_X'$, respectively.
All our experiments were conducted on a laptop computer with an Intel 4-core 2.8 GHz CPU and 64 GB memory. The algorithms are written in Java programming language and are implemented single threaded.

These experiments are reported in Table~\ref{tab:acc_size_time}. 
It is easily seen that the empirical error is far lower
than the error guarantee $\epsilon$ (around half),
even though we used the simple heuristic for the initial centers.
Halving $\epsilon$ typically doubles the coreset size,
but overall the coreset size is rather small,
and translates to a massive speedup (more than 500x)
in the time it takes to compute the objective value.
Such small coresets open the door to running on the data set
less efficient but more accurate clustering algorithms.

\begin{table}[t]
	\centering
	\captionsetup{font=small}
	\small
	\caption{Comparing coresets constructed for varying $\epsilon$ (and the same $p = 0.1n$ and $k=2$). 
	}
	\label{tab:acc_size_time}
	\begin{tabular}{ccccc}
		\toprule
		$\epsilon$ & emp.~err. & coreset size & $T_{\text{X}}$ (ms) & $T_{\text{X'}}$ (ms) \\
		\midrule
		50\% & 17.9\% & 122 & 143910 & 16 \\
		30\% & 14.3\% & 256 & 147216 & 15 \\
		20\% & 10.6\% & 475 & 131718 & 16 \\
		10\% & 7.0\% & 1603 & 134512 & 63 \\
		5\% & 2.8\% & 5385 & 130633 & 203 \\
		\bottomrule
	\end{tabular}
\end{table}

In Theorem~\ref{thm:okm}, making the coreset work for all $p$ values
incurs an $O(\log^2 n)$ factor in the coreset size (see Section~\ref{sec:simul}).
We thus experimented whether a single coreset $X'$,
that is constructed for parameters $p = 0.1 n$, $\eps = 10\%$, and $k = 2$,
is effective for a wide range of values of $p'\neq p$.
As seen in Table~\ref{tab:varying_p},
this single coreset achieves low empirical errors (without increasing the size).
We further evaluate this same coreset $X'$ (with $p = 0.1n$)
for weight vectors $w$ that satisfy a power law (instead of 0/1 vectors).
In particular, we let $w_i = \frac{1}{i^\alpha}$ for $\alpha > 0$,
and experiment with varying $\alpha$.
The empirical errors of this coreset, reported in Table~\ref{tab:weighted},
are worse than that in Table~\ref{tab:varying_p} and is sometimes slightly larger than the error guarantee $\eps = 10\%$, but it is still well under control. Thus, $X'$ serves as a simultaneous coreset for various weight vectors, and can be particularly useful in the important scenario of data exploration, where different weight parameters are experimented with.

\begin{table}[t]
	\centering
	\captionsetup{font=small}
	\small
	\caption{Evaluating a single coreset (constructed for $\eps = 10\%$, $p = 0.1n$, $k = 2$) for varying $p'$ and for varying power-law weights. 
	}
	\label{tab:two}
	\begin{subtable}{0.4\linewidth}
		\centering
		\caption{varying $p'$}
		\begin{tabular}{cc}
			\toprule
			$p'$ & emp.~err. \\
			\midrule
			$0.01n$ & 4.0\% \\
			$0.05n$ & 6.6\%\\
			$0.2n$  & 5.0\% \\
			$0.3n$  & 4.1\% \\
			$0.4n$	& 3.6\% \\
			$0.5n$  & 3.3\% \\
			$n$     & 4.5\% \\
		\end{tabular}
		\label{tab:varying_p}
	\end{subtable}
	\begin{subtable}{0.4\linewidth}
		\centering
		\caption{power-law weights}
		\begin{tabular}{cc}
			\toprule
			$\alpha$ & emp.~err. \\
			\midrule
			0.5 & 3.2\% \\
			1.0 & 9.0\% \\
			1.5 & 11.1\% \\
			2.0 & 11.5\% \\
			2.5 & 11.6\% \\
			3.0 & 11.7\% \\
			3.5 & 11.7\% \\
		\end{tabular}
		\label{tab:weighted}
	\end{subtable}

    \hrulefill
\end{table}

\ifprocs
\bibliographystyle{icml2019}
\else
\bibliographystyle{alphaurlinit}
\fi

\bibliography{probdb}

\end{document}